\documentclass[pra,twocolumn,preprintnumbers]{revtex4}

\usepackage{graphicx}
\usepackage{bm}
\usepackage{amsmath,hyperref}
\usepackage{amssymb}
\usepackage{amsfonts}
\usepackage{fancyhdr}
\usepackage{slashed}
\usepackage{latexsym,bbm}
\usepackage{amsthm}

\usepackage{algorithm,algpseudocode}

\usepackage{enumerate}

\newcommand{\bra}[1]{\left \langle{#1} \right |}
\newcommand{\ket}[1]{\left |{#1} \right \rangle}

\usepackage{color}
\definecolor{blue}{rgb}{0,0.2,1}

\definecolor{red}{rgb}{0.9,0,0}

\newcommand{\Ord}[1]{\mathcal{O}\left( #1 \right)}
\newcommand{\tOrd}[1]{\widetilde{\mathcal{O}}\left( #1 \right)}
\newcommand{\Tht}[1]{\Theta \left( #1 \right)}
\newcommand{\Omg}[1]{\Omega \left( #1 \right)}

\theoremstyle{plain}
\newtheorem{oracle}{Data Input}

\newtheorem{theorem}{Theorem}

\newtheorem{lemma}{Lemma}

\newtheorem{cor}{Corollary}
\newtheorem{fact}{Fact}

\def\be{\begin{eqnarray}}
\def\ee{\end{eqnarray}}

\usepackage{color}
\definecolor{Pr}{rgb}{0.4,0.3,0.9}

\begin{document}

\title{Quantum algorithms for hedging and the learning of Ising models}
\author{Patrick Rebentrost}
\thanks{Centre for Quantum Technologies, National University of Singapore, Singapore 117543}
\email{cqtfpr@nus.edu.sg}
\author{Yassine Hamoudi}
\thanks{Universit\'e de Paris, IRIF, CNRS, F-75013 Paris, France.}
\email{yassine.hamoudi@irif.fr}
\author{Maharshi Ray}
\thanks{Centre for Quantum Technologies, National University of Singapore, Singapore 117543.}
\author{Xin Wang}
\thanks{Institute for Quantum Computing, Baidu Research, Beijing 100193, China.}
\author{Siyi Yang}
\thanks{Centre for Quantum Technologies, National University of Singapore, Singapore 117543.}
\author{Miklos Santha}
\thanks{Universit\'e de Paris, IRIF, CNRS, F-75013 Paris, France;  and Centre for Quantum Technologies and  MajuLab UMI 3654, National University of Singapore, Singapore 117543.}
\email{cqtms@nus.edu.sg}

\date{\today}
\begin{abstract}
A paradigmatic algorithm for online learning is the Hedge algorithm by Freund and Schapire.
An allocation into different strategies is chosen for multiple rounds and each round incurs corresponding losses for each strategy. The algorithm obtains a favorable guarantee for the total losses even in an adversarial situation.
This work presents quantum algorithms for such online learning in an oracular setting. For $T$ time steps and $N$ strategies, we exhibit run times of about $\Ord{{\rm poly} (T) \sqrt{N}}$ for  estimating the losses and for betting on individual strategies by sampling.
In addition, we discuss a quantum analogue of the Sparsitron, a machine learning algorithm based on the Hedge algorithm. The quantum algorithm inherits the provable learning guarantees from the classical algorithm and exhibits polynomial speedups.
The speedups may find relevance in finance, for example for hedging risks, and machine learning, for example for learning generalized linear models or Ising models.
\end{abstract}

\maketitle

\section {Introduction}

Optimization is a cornerstone of machine learning and artificial intelligence.  Examples are the training of support vector machines and neural networks for tasks such as the analysis of audio, texts, and images.
A great deal of quantum algorithm developments have been concerned with quantum speedups for optimization problems, and in particular for convex optimization problems. Generic convex optimization in the quantum oracle model was discussed in \cite{Chakrabarti2018,Apeldoorn2018}.
Special cases of convex programming are linear programs (LP) and semidefinite programs (SDP), which involve optimizing a linear function of a vector or matrix subject to linear contraints, respectively.
Several quantum algorithms have been discussed for these two convex programs using amplitude amplification and estimation and Gibbs sampling \cite{Brassard2002,AK16, BS17,vAG18}.
Linear programs can also be mapped to zero-sum games, for which a linear time classical solver was developed by Grigoriadis and Khachiyan~\cite{GK95}.
For zero-sum games,  quantum algorithms were obtained in~\cite{vAG19,LCW19}, which in turn also apply to LPs.
Beyond LPs and SDPs, sublinear algorithms for quadratic constraints
are discussed by Clarkson, Hazan and Woodruff~\cite{CHW12},
which includes for example the classification of data points with a margin, a kernel-based classification, the minimum enclosing ball problem, and $\ell_2$-margin support vector machines.
Reference~\cite{LCW19} provides corresponding quantum algorithms based on the amplitude amplification and estimation subroutines.

These quantum optimization algorithms assume a quantum oracle that can be queried in superposition. Such a setting dates back to early quantum  algorithms such as Grover's or Deutsch-Jozsa's. In machine learning, data are usually generated from an external source, such as users providing ratings to movies or products. In this case, quantum random access memory (QRAM) is discussed as a way to make such data available to a quantum algorithm \cite{giovannetti2008quantum,Lloyd2013}.
The oracle framework used in this work encompasses such a QRAM data access and also the case of access to a computable function. For a quantum algorithm's output, a possible way is to encode the output in a quantum state. A famous example is the work of Harrow, Hassidim and Lloyd (HHL)~\cite{HHL09} for solving linear system of equations. The solution to a linear system $Ax=b$ is provided by a quantum state $\ket x$ upon which measurements can be performed to obtain classically relevant information. Under some well-discussed conditions~\cite{HHL09,Aar15}, the HHL algorithm can achieve an exponential quantum speedup.
In contrast,  in the optimization works mentioned above and in the present work, the output of the algorithm is inherently classical. These algorithms are hybrid, that is partly of classical and partly of quantum nature, and designed in a modular way so that the quantum part of the algorithm can be treated as a separate building block.
The algorithms make quantum improvements on parts of the (best) available classical algorithm while keeping its overall structure intact.
In contrast to the HHL algorithm for example, here the quantum versus classical speedup is usually at most polynomial, in most cases at most quadratic in the domain size of the input function. The quantum algorithm can deliver some speedup in the dimension of the problem, while in other relevant parameters it might not necessarily achieve any speedup, sometimes it can even be worse than the best classical algorithm.

Consider a game with $T$ rounds, where we have the chance to play a mixture of $N$ different strategies at each round and observe the results of our choice in the next round.
The setting is ``online" in the sense that the results are unknown ahead of time, which also can be seen as an idealized version of  sports betting or stock market trading.
The Hedge algorithm by Freund and Schapire adaptively changes the mixture of strategies (a probability vector) via multiplicative weight updates \cite{Freund1997}.
This strategy allows for  losses after $T$ rounds that are not much worse than the minimum achievable ``offline" loss. Here, ``offline" means that the strategy is picked in advance without any adaptation. This difference of online and minimum achievable offline loss is often called ``regret", a slight misnomer as in the online setting it is impossible to compute the minimum offline loss in advance.
Precisely, it can be shown that the regret of the Hedge algorithm is not
worse than $\sqrt{2 T \log N} + \log N$. The classical complexity for this algorithm is $\Ord{T N}$: at each round we have to perform the multiplicative update, an effort that is proportional to $N$. While $T$ and $N$ are arbitrary, in most applications when applied in learning theory for example, $T$ is much smaller than $N$, e.g., $T= \Ord{\log N}$.

In this work, we provide quantum algorithms in the online learning/hedging scenario.
Assuming appropriate oracles for the online loss information, we first exhibit quantum speedups for two settings, which can be considered the passive and active setting. In the passive setting, we are interested in estimating the total loss after $T$ rounds without ever writing down the full probability vector and without making active decisions.
We obtain an $\epsilon$-accurate estimate of the total loss with high probability and with
a query complexity of
$\Ord{\frac{T^{2} \sqrt{N}}{ \epsilon}}$
and a gate complexity of $\tOrd{\frac{T^{2}\sqrt{N}}{\epsilon}}$ via amplitude estimation.
We use the notation $\tOrd{}$ to hide poly-logarithmic factors in any of the variables.
The improvement in $N$ and worsening in $T$ is acceptable in the common situation when $T$ is much smaller than $N$.
Furthermore, we consider the active setting where at each round the strategy is executed and incurs a transaction cost. For this active setting we prepare the relevant quantum state with amplitude amplification, sample from it, and bet on the outcome. In the case when every allocation into a particular strategy comes at a transaction cost, the sampling setting has the advantage of reducing such costs.
We again obtain a quantum speedup in $N$ while the loss of such a strategy remains close to the minimum loss with high probability.
We note work on quantum speedups for the Hedge algorithm and the related adaptive boosting technique in \cite{Wang2019_boost, Arunachalam2020}.

The main motivation for discussing the Hedge algorithm is its application in optimization and machine learning. In the second part of this work, we provide a quantum algorithm for the Sparsitron \cite{Klivans2017}. The Sparsitron is a supervised learning algorithm based on the multiplicative weights algorithm by Freund and Schapire. The algorithm can be used to learn a Generalized Linear Model (GLM) and Ising models from training examples. For each $N$-dimensional training example, a loss is computed which takes into account a non-linear, potentially \textit{non-convex}, activation function and an inner product between the training example and a weight vector. For guaranteed learning from the data with a certain accuracy $\epsilon$ (to be defined), the run time is about
$\tOrd{\frac{N}{\epsilon^4 }}$.
We present a classical approximate Sparsitron algorithm, which estimates the inner products instead of computing them exactly (Theorem \ref{theoremSparsitronApprox}). The runtime of this algorithm is  $\tOrd{\frac{N}{\epsilon^2 } + \frac{1}{\epsilon^6}}$, improving on the original algorithm if, say,  $N >\frac{1}{\epsilon^2}$.
Subsequently, we present a quantum algorithm called the Quantum Sparsitron (Theorem \ref{theoremQuantumSparsitron}).
The quantum algorithm uses quantum inner product estimation and achieves a run time of about
$\tOrd{\frac{\sqrt{N}}{\epsilon^7 }}$,  a polynomial quantum speedup compared to the classical algorithm with respect to the dimension of the data.
As a corollary, we derive a polynomial quantum speedup for the learning of Ising models (Corollary \ref{corollaryIsing}).

Regarding notation, we use $[N]$ to denote the set $\{1,\dots,N\}$, where $N \in \mathbbm Z_+$. We write a vector plainly as $x$ without any special furnishing, however we use $\vec 0$ and $\vec 1$ to denote the all $0$s and all $1$s vector, respectively.
The $\ell_1$-norm of a vector $x \in \mathbbm R^N$ is given by
$\Vert x \Vert_1 := \sum_{j=1}^N \left \vert x_j \right \vert$.
The maximum element of a vector $x \in \mathbbm R^N$ is given by
$\Vert x \Vert_{\max} := \max_{j\in [N]} \left \vert x_j \right \vert$, also sometimes denoted by  $x_{\max}$.
Equivalently, the maximum absolute element of a matrix $A$ is denoted by $\Vert A \Vert_{\max}$.
For $x,y \in \mathbbm R^N$, we write the inner product as $x \cdot y$ and the element-wise vector multiplication as $x  \odot y \in \mathbbm R^N$, where $(x  \odot y)_j = x_j y_j$. For $c \in \mathbbm R$ and  $x \in \mathbbm R^N$, $c^x \in \mathbbm R^N$ is understood element-wise as $(c^x)_j = c^{x_j}$.
We use $\ket{\bar 0}$ to denote the multi-qubit state $\ket{0}\otimes \dots \otimes \ket{0}$, where the number of qubits is clear from the context.
As mentioned, we use the notation $\tOrd{}$ to hide poly-logarithmic factors in any of the variables.

\section{Classical Hedge algorithms}

\subsection{Original algorithm}

We follow Ref.~\cite{Freund1997} for the discussion of the classical Hedge algorithm. We are given $N$ strategies for a game that takes $T$ rounds. Before each time $t\in[T]$, we choose an assignment (portfolio) of the $N$ strategies. This assignment shall be given by the weights
$
w^{(t)} = \left (w_1^{(t)},\dots, w_N^{(t)} \right ) ^\dagger \in \mathbbm [0,1]^N,
$
which form the probability vector
\be \label{eqProbabilities}
p^{(t)} = \left (p_1^{(t)},\dots, p_N^{(t)} \right ) ^\dagger = \frac{1}{\Vert w^{(t)} \Vert_1} \left (w_1^{(t)},\dots, w_N^{(t)} \right) ^\dagger.
\ee
The initial allocation is taken to be uniform, i.e.,
$
w^{(1)} = \left (1/N,\dots, 1/N\right)^\dagger$ and $p^{(1)} = \left (1/N,\dots, 1/N \right)^\dagger.
$
The algorithm considers an online learning setting, where information arrives over time and the weights are updated accordingly. Specifically, at each time $t\in[T]$, we observe the loss vector
\be
l^{(t)} = \left (l_1^{(t)},\dots, l_N^{(t)} \right) ^\dagger \in  [0,1]^N.
\ee
Algorithmically, we describe this as ``Receive loss vector $l^{(t)}$", which means we obtain  access to the loss vector.
To avoid further complexities, we assume that each loss $l^{(t)}_j$ takes a constant number of bits to specify.
The loss at time $t $ is given by
\be \label{eqStepLoss}
L^{(t)} := \sum_{i=1}^N p_i ^{(t)} l_i ^{(t)} \equiv p ^{(t)} \cdot l ^{(t)} \in [0,1].
\ee
Algorithmically,  this loss is taken into account with the statement ``Suffer loss $L^{(t)}$", which means we add $L^{(t)}$ to the overall loss amount.
A strategy to minimize losses was shown in Ref.~\cite{Freund1997}. Take $\beta \in (0,1)$. The strategy is based on multiplicative updates to the weights given the incoming loss information as
$w_j^{(t)} = \beta^{l^{(t-1)}_j} w_j^{(t-1)}$,
which for the full path up to $t$ is
$w_j^{(t)} = \beta^{\sum_{t'=1}^{t-1} l^{(t')}_j} w_j^{(1)}$.
We also write $w^{(t)}  \odot \beta^{l^{(t)}}$, using the notation for the element-wise vector multiplication, and $\beta^{l^{(t)}}$ is understood element-wise.

The original algorithm is given in Algorithm \ref{algSchapire}, setting the
flag to ``null".
The accumulated loss of this algorithm (denoted by $\mathcal H$ for hedge) over $T$ rounds is
\be
L_{ \mathcal H} := \sum_{t=1}^T  L^{(t)}.
\ee
On the other hand, consider the ``offline loss", $L^{(T)}_{\min} = \min_{j\in[N]}   \sum_{t=1}^T l_j^{(t)}$,
which gives the minimum loss achievable when choosing the same single strategy for all rounds of the game. 
Ref.~\cite{Freund1997} shows a ``regret" bound for the losses of the multiplicative update strategy.
\begin{theorem}[Regret bound \cite{Freund1997}]\label{thmRegret}
With the definitions above and $\beta = 1/(1+\sqrt{2 \log N /T})$, Algorithm \ref{algSchapire} with flag ``null'' achieves
\be
L_{ \mathcal H} - L^{(T)}_{\min} \leq \sqrt{2 T \log N} + \log N.
\ee
\end{theorem}
This bound is better than the naive bound $L_{ \mathcal H} -L^{(T)}_{\min}  \leq T$.
The run time of the classical algorithm is given by
$\Ord{T N}$.
At every step, the algorithm updates $N$ probabilities and there are $T$ steps overall.

A quantum version of this algorithm with the appropriate oracles for the loss vectors will be able to provide an estimate of the total loss $L_{ \mathcal H}$ with polynomial speedup in $N$, see Algorithm \ref{algQuantumEstimateLoss} below.
The ``deterministic" and ``sampled" flags will be explained in the next subsection.
\begin{figure}
\begin{algorithm}[H]
  \caption{Hedge algorithm by Freund and Schapire \cite{Freund1997} with transaction cost and sampling}
    \label{algSchapire}
  \begin{algorithmic}[1]
    \Require{Number of strategies $N$, number of rounds $T$, parameter $\beta \in (0,1)$, flag $\in \{$null, deterministic, sampled$\}$, fixed transaction cost per strategy $C_0$.}
    \State $w^{(1)} \gets   \vec 1/N$.
    \For{$t \gets 1 \textrm{ to } T$}
    \State $p^{(t)} \gets w^{(t)}/\Vert w^{(t)} \Vert_1$. \Comment Update probabilities
     \If{flag = null}
     \State $q^{(t)} \gets p^{(t)} $.
      \State $C^{(t)}  \gets  0$. \Comment No transaction cost
     \ElsIf{flag = deterministic}
     \State $q^{(t)} \gets p^{(t)} $.
     \State Allocate portfolio according to $q^{(t)}$.
    \State $C^{(t)} \gets  N \times C_0$.
    \ElsIf{flag = sampled}
       \State Prepare sampling data structure for $p^{(t)}$ via Fact \ref{theoremSampling}.
    \State $j \gets$ Sample from $p^{(t)}$.
    \State $q^{(t)} \gets \vec 0$, $q^{(t)}_j \gets 1$.
    \State Allocate portfolio according to $q^{(t)}$.
    \State $C^{(t)}  \gets  C_0$.
    \EndIf
    \State Receive loss vector $l^{(t)}$.
    \State Suffer loss $L^{(t)} \gets q^{(t)} \cdot l^{(t)}$.
    \State $w^{(t+1)} \gets w^{(t)}  \odot \beta^{l^{(t)}}$.  \Comment Multiplicative weight update
    \EndFor
     \Ensure { $\sum_{t=1}^T L^{(t)}$, $\sum_{t=1}^T  C^{(t)}$.
  }
  \end{algorithmic}
\end{algorithm}
\end{figure}
\subsection{Active betting by sampling}

In a simple extension to the original algorithm, we model the cost of allocating individual bets in the portfolio.
This allocation/transaction cost serves to illustrate the benefits of a sampling strategy over the deterministic strategy.
The allocation costs are taken to be as follows. Each allocation based on the elements of the vector $p^{(t)}$, however small, will come with a cost $C_0$. For example, buying even a single stock on the stock market comes with the cost of contacting a broker and a broker fee.
This cost is counted separately for the algorithm in a total transaction cost $C$. We keep these costs separate from the losses $ L_{\mathcal H}$, and it is straightforward to combine them.
Algorithm \ref{algSchapire} formalizes this additional feature.
With ``Allocate portfolio" we denote the step of concretely betting on respective strategies which incurs a cost of $C_0$ per strategy allocated. The flag ``deterministic" leads to an allocation into all the strategies with corresponding transaction cost.
To minimize the transaction cost, the flag ``sampled" can be used, for which at each time $t$  an index $j^{(t)}$ is sampled according to $p^{(t)}$ and an investment is made on that single outcome. This sampling requires preparing a sampling data structure, as explained in Fact \ref{theoremSampling} in Appendix \ref{appendixSampling}. In the quantum case, see Algorithm \ref{algorithmHedgeSampleQuantum}, a quantum state $\ket{p^{(t)}}$ is prepared from which samples are obtained.

The run time of the algorithm is the same for all three flags.
A simple statement about the total transaction cost is as follows.
\begin{fact}  \label{thmRuntimeTransactionCost}
The transaction cost $C:=\sum_{t=1}^T  C^{(t)}$ of Algorithm \ref{algSchapire} with the different flags is
\be
{\rm null } &:& C= 0. \\
{\rm deterministic } &:& C = N \times  T \times C_0 , \\
{\rm sampled }&:& C =  T  \times C_0.
\ee
\end{fact}
\begin{proof}
For ``deterministic",  we have $T$ iterations, at each of which the investment cost is $N \times C_0$.
For ``sampled", at each step we prepare the data structure according to Fact \ref{theoremSampling}, sample, and invest in the single sampled strategy.  The transaction cost per step is $C_0$, independent of $N$.
\end{proof}
Furthermore, the output of the sampled algorithm is an unbiased estimator of $L_{\mathcal H}$ and satisfies a regret bound.
\begin{theorem} \label{theoremHedgeCorrectnessSampling}
Algorithm \ref{algSchapire} with flag ``sampled" and $\beta = 1/(1+\sqrt{2 \log N /T})$  outputs $L_{\rm samp} := \sum_{t=1}^T l_{j^{(t)}}^{(t)}$ such that $\mathbbm E[L_{\rm samp} ] =  L_{\mathcal H}$. For $\delta \in (0,1)$, the regret bound is
$L_{\rm samp} - \min_{j\in[N]} L_j \leq 3\sqrt{T \log (N/\delta)} +\log N$ with probability at least $1- \delta$.
\end{theorem}
\begin{proof}
The expectation value of $L_{\rm samp}$ is
\be
\mathbbm E\left [L_{\rm samp} \right] &=&   \sum_{t=1}^T  \mathbbm E \left [l_{j^{(t)}}^{(t)}\right ] = \sum_{t=1}^T p^{(t)} \cdot l^{(t)} \equiv L_{ \mathcal H}.
\ee
Hoeffding's inequality is $P\left [ \frac{L_{\rm samp}}{T} - \frac{L_{ \mathcal H}}{T} \geq r \right] \leq  e^{-2 T r^2 }$
 for $r \in \mathbbm R_+$ and  hence
\be
P\left [  L_{\rm samp} -  L_{ \mathcal H} \geq s \sqrt{T} \right] \leq  e^{-2 s^2},
\ee using $r = s / \sqrt T$.
Now, bound the difference to the minimum loss using Theorem \ref{thmRegret} as
\be L_{\rm samp} - \min_{j\in[N]}  L_j &=& L_{\rm samp} -  L_{ \mathcal H} +  L_{ \mathcal H} - \min_{j\in[N]}  L_j \nonumber \\
&\leq& L_{\rm samp} - L_{ \mathcal H} + \sqrt{2T \log N} +\log N \nonumber \\
&\leq& s \sqrt{T}+ \sqrt{2T \log N} +\log N,
\ee
with probability at least $1-e^{-2 s^2}$.
Setting $s = \sqrt{ \ln(1/\delta) /2} \leq  \sqrt{ \log( N/\delta)} $ leads to
\be
 L_{\rm samp} - \min_{j\in[N]}  L_j \leq  3 \sqrt{T \log (N/\delta)} +\log N
\ee
with probability at least $1- \delta$.
\end{proof}
Hence the sampling Hedge algorithm achieves the correct loss in expectation and the regret bound up to a scaling factor with high probability.
\section{Quantum Hedge algorithms}

We now turn to quantum algorithms in the Hedge setting. We provide two simple algorithms, one for estimating the losses, one for the active betting scenario, before discussing the Sparsitron.
The algorithms are based on quantum minimum finding, see Lemma \ref{lemmaMin}, amplitude amplification and estimation, see Lemma \ref{lemmaU}, and quantum inner product estimation, see Lemma \ref{lemmaInnerProduct1}, all in Appendix \ref{appendixAmp}. We first discuss the rescaling of relevant quantities. We then discuss the quantum data input model and state preparation subroutines. In a passive setting, we use amplitude estimation to estimate the total loss of the Hedge algorithm given the data input, see Algorithm \ref{algQuantumEstimateLoss}.
In an active setting, we discuss the allocation of a portfolio via amplitude amplification and sampling, see Algorithm \ref{algorithmHedgeSampleQuantum}.

One of the quantities estimated via a quantum algorithm is the  $\ell_1$-norm $\left \Vert w^{(t)} \right \Vert_1$ for all $t$. Recall that the algorithm starts  with a uniform initial weight vector $w^{(1)} = \vec 1/N$. As we are decreasing each weight by at most $\beta$ per step, the minimum weight achievable after $T$ steps is
${\beta^T/N}$.
In a situation where also the maximum weight $w^{(T)}_{\max} \equiv \left \Vert w^{(T)} \right \Vert_{\max}$ is $\Ord{\beta^T/N}$, the $\ell_1$-norm is small, i.e., $\left \Vert w^{(T)} \right\Vert_1 \sim \beta^T\sim 1/2^T$.  To overcome this inconvenient worst case, we rescale the weights in the estimation. For each $1 \leq t \leq T$, define similar to before the minimum offline loss up to $t$,
$L^{(t)}_{\min} := \min_{j\in[N]}  \sum_{t'=1}^{t} l_j^{(t')}  \leq t$.
Note that the maximum element of $w^{(t)}$
is
\be \label{eqWeightMax}
w^{(t)}_{\max} =  \beta^{L^{(t-1)}_{\min}}/N.
\ee
We consider the rescaled weights $\frac{w^{(t)}_j} {w^{(t)}_{\max} } \leq 1$,
which keep the expected loss $L^{(t)} \equiv  \frac{w^{(t)} \cdot l ^{(t)}}{ \left \Vert w^{(t)}\right \Vert_1}$ the same because the $w^{(t)}_{\max}$ factor cancels out. However, we now have the lower bound for the $\ell_1$-norm
$\left \Vert \frac{w^{(t)}}{w^{(t)}_{\max} } \right \Vert_1 \geq 1$.
In the quantum context, we find $L^{(t-1)}_{\min}$ via the minimum finding algorithm \cite{Durr1996} in run time $\tOrd{\sqrt{N}}$ with high success probability. See Lemma \ref{lemmaMin} in Appendix \ref{appendixAmp} for the statement of the minimum finding algorithm. With Eq.~(\ref{eqWeightMax}) we can then compute the maximum weight.

\subsection{Quantum data input model}

We translate the online learning setting into the quantum domain. The input data for the quantum algorithms are the losses experienced at every step $t$. First, we assume $T$ different oracles, where the sequential access to these oracles embodies the online setting.
\begin{oracle}[Loss oracles]\label{oracle1}
Assume that $\Ord{1}$ bits are sufficient to specify the losses $l_j^{(t)}$.
For $t\in[T]$ and $j\in[N]$, assume unitaries $U_{l^{(t)}}$  such that $U_{l^{(t)}} \ket j \ket {\bar 0} = \ket j \ket {l^{(t)}_j}$, operating on $\Ord{\log N}$ quantum bits.
\end{oracle}
To simplify the notation, we have used the $\ket{\bar 0}$ initial state in the second register. In fact, the precise assumption here is $U_{l^{(t)}} \ket j \ket {c} = \ket j \ket {c\oplus l^{(t)}_j}$, where $c$ is any bit string allowed in the second register and $\oplus$ is the bit-wise addition modulo $2$. With this precise assumption, we have $\left (U_{l^{(t)}} \right )^2 \ket j \ket {c} = \ket j \ket {c} $, hence we can uncompute the result with another query.
This assumption shall hold also for the other data inputs below.
The oracles allow us to perform the following computations.
\begin{lemma} \label{lemmaUpdateBit}
Let $t\in [T]$ and  $\beta\in (0,1)$.
Let the set of unitaries $U_{l^{(t')}}$ for $t' \in [t-1]$ be given as in Data Input \ref{oracle1}. There exists unitaries performing the computations $\ket j \ket{\bar 0} \to \ket j \ket{ \sum_{t'=1}^{t-1} l_j^{(t')} }$, $\ket j \ket{\bar 0} \to \ket j \ket{ w^{(t)}_j}$, and, if knowledge of $w_{\max}^{(t)}$ is given, $\ket j \ket{\bar 0} \to \ket j \ket{  \frac{w^{(t)}_j}{w^{(t)}_{\max}}}$ to sufficient accuracy $\Ord{1/N}$. 
These computations take $\Ord{T}$ queries to the data input and  $\Ord{T+\log N}$ qubits and quantum gates. 
\end{lemma}
Hence, given the loss unitaries, we can compute the weights with overhead about $\Ord{T}$. This computation allows the estimation of norms and inner products and preparation of weight quantum states.

\subsection{Quantum algorithm to obtain the total loss of the Hedge algorithm}

Our first quantum algorithm is simple. At each time step, we receive a loss oracle according to Data Input \ref{oracle1}. From this input, we estimate the total loss of the multiplicative weight update method. We never fully exhibit the full weight vector but rather only the total loss at each step given the weight vectors.
The quantum algorithm is given in Algorithm \ref{algQuantumEstimateLoss}.
\begin{figure}
\begin{algorithm}[H]
  \caption{Quantum estimation of the total loss of the Hedge algorithm}
    \label{algQuantumEstimateLoss}
  \begin{algorithmic}[1]
    \Require{Number of strategies $N$, number of rounds $T$, parameter $\beta \in (0,1)$, error $\epsilon\in (0, 1]$, success probability $1-\delta \in (0,1)$.}
    \For{$t = 1 \textrm{ to } T$}
    \State Construct unitary for $w^{(t)} = w^{(1)} \beta^{l^{(1)}}\cdots \beta^{l^{(t-1)}}$ using oracles $\{U_{l^{(t')}} : t' \in [t-1] \}$ and Lemma \ref{lemmaUpdateBit}.
    \State Receive loss oracle $U_{l^{(t)}}$.
    \State $L^{(t)} \gets$ Quantum estimate $\frac{w^{(t)}}{\Vert w^{(t)}\Vert_1} \cdot l^{(t)}$ to relative accuracy ${\epsilon}$ with success probability $1-\frac{\delta}{T}$ via Lemma \ref{lemmaInnerProduct1}.
      \EndFor
     \Ensure {$\sum_{t=1}^T {L^{(t)} }$.}
  \end{algorithmic}
\end{algorithm}
\end{figure}

We use the Lemma \ref{lemmaInnerProduct1} for estimating the inner product with relative accuracy and hence obtain a relative accuracy for the total loss.
We obtain a statement on the accuracy of the loss, the run time, and the success probability of the quantum algorithm.
\begin{theorem} \label{theoremQuantumEstimate}
Let $\epsilon \in (0, 1]$, $\delta \in (0,1)$, and $\beta = (0,1)$. Algorithm \ref{algQuantumEstimateLoss} provides an estimate $\sum_{t=1}^T  L^{(t)}$ of the total loss $ L_{\mathcal H}$ such
that
$\left \vert \sum_{t=1}^T  L^{(t)} - L_{ \mathcal H}\right \vert \leq \epsilon L_{ \mathcal H}$ with probability at least $1-\delta$.
This quantum algorithm requires $\Ord{\frac{T^2 \sqrt{N}}{ \epsilon} \log \left(\frac{T}{\delta}\right)}$ queries to the oracles and $\tOrd{\frac{T^2\sqrt{N}}{\epsilon} \log \left(\frac{1}{\delta}\right)}$ gates.
\end{theorem}
\begin{proof}
Regarding the correctness, assuming all subroutines succeed,
we have for all $t\in[T]$ that
$\left \vert  L^{(t)} - \frac{w^{(t)}}{\Vert w^{(t)}\Vert_1} \cdot l^{(t)} \right \vert \leq  \epsilon \frac{w^{(t)}}{\Vert w^{(t)}\Vert_1} \cdot l^{(t)}$,
and thus
$\left \vert \sum_{t=1}^T  L^{(t)} - L_{ \mathcal H} \right \vert \leq \epsilon L_{ \mathcal H}$.

The run times for computing the entries $w_j$ are given via Lemma \ref{lemmaUpdateBit},  amounting to
$\Ord{T}$ queries to the Data Input \ref{oracle1} and  $\Ord{T+  \log (TN/\epsilon)}$ quantum gates.
Lemma \ref{lemmaInnerProduct1} at each step requires
$\Ord{\frac{ \sqrt{N} }{\epsilon} \log \left (\frac{T}{\delta} \right )  }$ queries and $\tOrd{\frac{ \sqrt{N} }{\epsilon} \log \left (\frac{T}{\delta} \right )  }$ quantum gates. These complexities are multiplied by $T$ as we are proceeding for $T$ steps. From Lemma \ref{lemmaUpdateBit}, the complexities are multiplied by another factor $\Ord{T}$ for the query complexity and a factor $\Ord{T+  \log (TN/\epsilon)}$ for the gate complexity.
Each single-step, single-estimate success probability is  $1-\frac{\delta}{T}$.
Hence the overall success probability of the algorithm is $\left (1-\frac{\delta}{T}\right)^{T} \geq 1-\delta$.
\end{proof}

\subsection{Active betting with quantum sampling}

Next, we provide a quantum version of Algorithm \ref{algSchapire} with the ``sampled" flag. Instead of obtaining the sample classically, we prepare a corresponding quantum state and measure it.
We approximately prepare the quantum states of square-root probabilities
$\ket{ p^{(t)} } = \sum_{j=1}^N \sqrt{ p_j^{(t)}} \ket j$,
for every time $t=1,\dots,T$, with the probabilities given in Eq.~(\ref{eqProbabilities}).
The quantum algorithm is given in Algorithm \ref{algorithmHedgeSampleQuantum}.
We obtain a biased estimate of the loss where the bias is set to $\sqrt{T\log N}$ to obtain the usual regret bound up to constant factors. 
\begin{figure}
\begin{algorithm}[H]
  \caption{Active Hedge algorithm with transaction cost and quantum sampling}
    \label{algorithmHedgeSampleQuantum}
  \begin{algorithmic}[1]
    \Require{Number of strategies $N$, number of rounds $T$, parameter $\beta \in (0,1)$, transaction cost $C_0$, success probability $1-\delta \in(0,1)$.  }
    \For{$t = 1 \textrm{ to } T$}
    \State $L_{\min}^{(t-1)} \gets$ Find $\min_{j\in[N]}  \sum_{t'=1}^{t-1} l_j^{(t')}$ using oracles $\{U_{l^{(t')}} : {t'}\in [t-1] \}$ with success probability $1-\frac{\delta}{2T}$ using Lemma \ref{lemmaUpdateBit} and Lemma \ref{lemmaMin}.
    \State $w^{(t)}_{\max} \gets  \beta^{L^{(t-1)}_{\min}}/N$.
     \State Prepare an approximation $\ket {\tilde p^{(t)}}$ of the quantum state $\ket {p^{(t)}}$ with $p^{(t)} = w^{(t)}/\Vert w^{(t)}\Vert_1$ using $w^{(t)}_{\max}$,
     Lemma \ref{lemmaUpdateBit}, and Lemma \ref{lemmaU} (iii)  with accuracy $\xi = \sqrt{\log (N)/T}$ and  success probability $1-\frac{\delta}{2T}$.
      \State $j^{(t)}\gets$ Measure $\ket {\tilde p^{(t)}}$ in the computational basis.
      \State Allocate portfolio in $j^{(t)}$-th strategy at cost $C_0$.
      \State Receive loss oracle $U_{l^{(t)}}$.
      \State $l_{j^{(t)}}^{(t)}\gets$ Query $U_{l^{(t)}}$ with $\ket {j^{(t)}} \ket 0$.
      \State Suffer loss $l_{j^{(t)}}^{(t)}$.
    \EndFor
     \Ensure {$L^Q_{\rm samp} := \sum_{t=1}^T l_{j^{(t)}}^{(t)}$.}
  \end{algorithmic}
\end{algorithm}
\end{figure}

\begin{theorem}\label{theoremSampleQuantum}
Let $\delta \in (0,1)$. Algorithm \ref{algorithmHedgeSampleQuantum} outputs $ L^Q_{\rm samp} := \sum_{t=1}^T l_{j^{(t)}}^{(t)}$ with a bias $\left \vert \mathbbm E[L^Q_{\rm samp} ] -  L_{ \mathcal H} \right \vert \leq \sqrt{T\log N} $ with success probability at least $1-\delta$.
With $\beta = 1/(1+\sqrt{2 \log N /T})$, the regret bound is $L^Q_{\rm samp} - \min_{j\in[N]}  L_j \leq  4\sqrt{T \log (N/\delta)} + \log N$ with success probability at least $1-2\delta$. This quantum algorithm requires $\Ord{T^2 \sqrt{N} \log \left(\frac{T}{\delta}\right)}$ queries to the oracles and $\tOrd{T^2\sqrt{N} \log \left(\frac{1}{\delta}\right)}$ gates. The transaction cost is $T\times C_0$.
\end{theorem}
\begin{proof}
For every step $t\in[T]$, amplitude amplification produces an erroneous output state, see Lemma \ref{lemmaU}, with probabilities denoted by $\tilde p_j$.
Let $\xi \in (0,1]$ as in Lemma \ref{lemmaU} (iii), where by construction we have  $\left \Vert p^{(t)} - \tilde p^{(t)}\right \Vert_1 \leq \xi$.
The expectation value using these probabilities is
\be
\tilde L_{\mathcal H} := \mathbbm E_{\tilde p} \left [L^Q_{\rm samp} \right] = \sum_{t=1}^T  \tilde p^{(t)}  \cdot l^{(t)}.
\ee
Note that $\left \vert l^{(t)} \cdot \left( p^{(t)} - \tilde p^{(t)}\right) \right \vert \leq \left \Vert  p^{(t)} -  \tilde p^{(t)} \right \Vert_1 \left \Vert l^{(t)} \right \Vert_{\max} \leq \xi$.
Thus, we have the bias
\be
\left \vert \tilde L_{ \mathcal H} - L_{ \mathcal H} \right \vert \leq T \xi .
\ee
The success probability for obtaining this bias is $\left (1-\frac{\delta}{2T}\right)^{2T}  \geq (1-\delta)$.

For the difference to the minimum loss strategy, we find using Theorem \ref{thmRegret} that
\be
L^Q_{\rm samp} &-& \min_{j\in[N]}  L_j \leq  \nonumber \\ &\leq& L^Q_{\rm samp} -  \tilde L_{ \mathcal H}  + \left \vert \tilde L_{ \mathcal H} -L_{ \mathcal H} \right  \vert+L_{ \mathcal H} - \min_{j\in[N]}  L_j \nonumber \\
&\leq&  L^Q_{\rm samp} - \tilde L_{ \mathcal H}  + T \xi   + \sqrt{ 2 T \log N}  + \log N \nonumber\\
&\leq& T \xi + 3\sqrt{T \log (N/\delta)} + \log N.
\ee
Here, we used that $ L^Q_{\rm samp} - \tilde L_{ \mathcal H} \leq \sqrt{ T \log (N/\delta)} $ with probability $1-\delta$ from Hoeffding's inequality as in Theorem \ref{theoremHedgeCorrectnessSampling}. We have a total success probability for this regret bound of $(1-\delta)^2 \geq 1-2\delta$. Finally, set $\xi = \sqrt{\log(N)/T}$ to obtain the stated regret bound.

The minimum findings take a total run time of $\Ord{T^2 \sqrt N \log \left(\frac{T}{ \delta} \right)}$ via Lemma \ref{lemmaMin}.
  Lemma \ref{lemmaU} with $u_j = w_j^{(t)}/w^{(t)}_{\max}$ takes
 $\Ord{T^2 \sqrt{N} \log \left(\frac{T}{\delta}\right)}$ queries to the oracles and $\tOrd{T^2\sqrt{N} \log \left(\frac{1}{\delta}\right)}$ gates.
\end{proof}
We complete this section with comparison of the run time of Theorem \ref{theoremSampleQuantum} with the run time of the classical algorithm from Theorem  \ref{theoremHedgeCorrectnessSampling}.
Choose $\delta= \Tht{1}$. Then Theorem \ref{theoremSampleQuantum} achieves a regret bound of $\Ord{\sqrt{T \log N}}$ with a run time of $\tOrd{T^{2}\sqrt{N}}$ and high success probability.
The classical algorithm discussed in Theorem \ref{theoremHedgeCorrectnessSampling} achieves a similar regret bound $\Ord{\sqrt{T \log N}}$ with a run time of $\tOrd{TN}$
and high success probability.

\section{Sparsitron}
\label{sectionSparsitron}

At a high level, statistical classification in machine learning is performed with two main approaches. Let $X$ denote the observable variables (features) and $Y$ the target variables (labels). The first approach is using a \textit{discriminative model}. In this approach, one  models the conditional distribution $P\left[Y | X =x\right ]$, i.e., the probability of the labels given an input example $x$. On the other hand, a \textit{generative model}  constructs a joint distribution $P[X,Y]$ of variables and labels.  In this more general approach, inferences in both directions features $\to$ label and label $\to$ features can be made.

Undirected graphical models, or Markov random fields, are a powerful, modern statistical tool for modeling  high-dimensional  probability distributions \cite{Bresler2015,Klivans2017}. The joint probability distribution of such a model depends on an underlying graph, where the presence of an edge gives conditional dependence and the absence of an edge gives conditional independence.
The Ising model is a special type of Markov random field which uses binary variables and pairwise interactions.
Consider here $N$ binary variables $Z_j\in\{-1,1\}$ for $j\in[N]$. Associated with these variables is an undirected dependency graph. The graph enters a probability distribution
\be \label{eqProbGibbs}
P[Z=z] \propto \exp\left( \sum_{i,j:i\neq j} A_{ij} z_i z_j + \sum_i \theta_i z_i\right),
\ee
where $A\in \mathbbm R^{N\times N}$ is the graph adjacency matrix and $\theta \in \mathbbm R^N$ describes bias terms.
The width of an Ising model is defined as $\lambda(A,\theta)=\max_i \left(\sum_j \left\vert A_{ij} \right \vert + \vert \theta_i \vert\right)$.
An important task in a machine learning context is the following \textit{unsupervised} learning task:
Given samples from the distribution Eq.~(\ref{eqProbGibbs}) on $\{-1,1\}^N$,  learn the matrix $A$. This problem is also known as the inverse Ising problem, in contrast to modelling a physical system  (e.g., a spin system) and studying its behavior and effects such as phase transitions.

We provide a short overview of the hardness of the Ising model learning and the complexities of known algorithms. We refer the reader to \cite{Bresler2015, Klivans2017} as entry points into the literature, from which also the following overview derives.
An unconditional lower bound on the sample complexity of learning Ising models with $N$ vertices was given by Santhanam and Wainwright \cite{SW12}. Even if the weights of the underlying graph are known and greater than $\eta >0$,
any algorithm for learning the graph structure must use $\Omg{\frac{2^{\lambda(A,\theta) /4} \log N}{\eta 2^{3\eta} }}$ samples.
A famous early example of an efficient algorithm is due to Chow and Liu from 1968 \cite{CL68}  for when the underlying graph is a tree.
For learning Ising models on general graphs with $N$ nodes of degree at most $d$, the best known method is based on exhaustive search costing $N^d$.
Subsequent works addressed the inverse Ising model problem for restricted classes of graphs or restricted nature of interactions between the variables. Some of these include learning Ising models when the underlying graph structure is a polytree~\cite{polytree}, hypertree~\cite{hypertree}, tree mixture~\cite{mixtree}, or when the underlying graph has the correlation decaying property (which, informally means that two variables are asymptotically independent as the graph distance between them increases)~\cite{correlationdecay}.
With modest assumptions, Bresler gave a simple greedy algorithm \cite{Bresler2015} which reconstructs arbitrary Ising models on $N$ nodes of maximum degree $d$ in time $\tOrd{N^2}$ but has a sample complexity that depends doubly exponentially on $d$, whereas a singly exponential dependence on $d$ is necessary. He also gave evidence that this run-time is optimal by showing that beating the $\tOrd{N^2}$ run-time bound would imply an improved run-time algorithm for 
the well-studied \emph{light-bulb} problem~\cite{valiant88}.
Subsequently, the work of Bresler was subsumed by Klivans and Meka~\cite{Klivans2017}, who gave a $\Ord{N^2}$ algorithm with optimal sample complexity. Their main tool is a multiplicative weight update method algorithm (called \emph{Sparsitron}) for learning an hypothesis class called the \emph{generalized linear models} (discussed below). As an application of the Sparsitron algorithm, they show the Ising model can be learned efficiently in both time and sample complexity sense. We discuss their algorithm in more details below and eventually turn it into a quantum algorithm with the same sample complexity but with an improved run-time.
In terms of the nature of algorithms used, besides exhaustive searching and greedy algorithms, several works have used convex optimization techniques such as in~\cite{RLW10,LGK06}. A very recent work~\cite{daskalakis2020treestructured} gives a time and sample efficient algorithm to learn Ising models where the underlying graph is a tree, without making any further assumptions on the underlying distribution (such as bounds on the strengths of the model’s edges/hyper-edges).
In the quantum setting, to the best of our knowledge our work provides the first quantum algorithm for the inverse Ising model problem.

We continue with a more detailed discussion, which provides the connection to the following algorithm, the Sparsitron \cite{Klivans2017}. Focusing on a particular variable $Z_j$ in Eq.~(\ref{eqProbGibbs}), and setting the other variables to $x \in \{-1,1\}^{[N]\setminus \{j\}}$, we have for the conditional probability of the Ising model
\be
P[Z_j&=&-1 | Z_{\neq j} = x ] = \frac{P[Z_j=-1, Z_{\neq j} = x ]}{P[Z_{\neq j} = x ]} \nonumber
\\ &=&
\frac{P[Z_j=-1, Z_{\neq j} = x ]}{P[Z_j=-1, Z_{\neq j} = x ] + P[Z_j= 1, Z_{\neq j} = x ]} \nonumber
\\&=&
 \frac{1}{1+ \exp(4 \sum_{k\neq j} A_{jk} x_k +2  \theta_j) } \nonumber
 \\&=& \sigma(w \cdot x - 2\theta_j).
\ee
Here, we have used the sigmoid function $\sigma(z) = 1/(1+e^{-z})$ and the vector $w \in \mathbbm R^{[N]\setminus \{j\}}$ with $w_k = - 4 A_{jk}$.
This single-variable conditional probability suggests a way to turn the original unsupervised learning problem into a \textit{supervised} learning problem.
Setting $X = (Z_k, k\neq j)$ and $Y = (1-Z_j)/2$ turns all variables except the $j$-th one into features and the $j$-th variable into a label.
Note that $\mathbbm E[Y|X=x]=P[Z_j=-1 | Z_{\neq j} = x] = \sigma(w \cdot x -2 \theta_j)$.
In addition, we can redefine the vectors to absorb the bias term \cite{Klivans2017}. Increase the dimension of  $w$ by $1$ and change $x \leftarrow [x_1, \cdots,x_{N-1}, 1]^T$. This change straightforwardly  allows the learning of the bias $-2\theta_j$. We have the following problem statement for learning such a generalized linear model (GLM).
Given samples $(X,Y)$ from a distribution $\mathcal D$ with the conditional mean function $\mathbbm E[Y|X=x] = \sigma(w \cdot x)$, learn $w$.

Reference \cite{Klivans2017} developed the Sparsitron, an efficient classical method to learn GLMs based on the Hedge algorithm by Freund and Schapire. One assumption is that a true
$w$ with $\Vert w\Vert_1 \leq \lambda$ exists where $\lambda\geq 0$ is known. Without loss of generality one can take $w\geq 0$ and that $\Vert w\Vert_1 = \lambda$, see \cite{Klivans2017}.
A square loss function, or ``risk", for any vector $v \in [-1,1]^N$ is given by
\be
\varepsilon(v) := \mathbbm E_{(X,Y )\sim \mathcal D} \left [\left(\sigma\left (v\cdot X\right) - \sigma\left(w \cdot X\right)\right)^2\right].
\ee
The learning task here is defined in terms of finding a $v$ such that this risk is $\epsilon$ small with high probability.
The classical algorithm assumes access to two training sets, one with $T$ samples which will be used for constructing predictions ($\lambda p^{(t)}$ in the algorithm) and the other with $M$ samples for finding the best prediction among them. The second training set is used to evaluate the risk in an empirical manner, as we do not have access directly to the distribution $\mathcal D$. The ``empirical risk" for any vector $v \in [-1,1]^N$ and given samples $(a^{(m)},b^{(m)}) \in [-1,1]^N \times [0,1]$ for $m\in[M]$ from $\mathcal D$ is given by
\be \label{eqRiskEmp}
\widehat \varepsilon(v) := \frac{1}{M} \sum_{m=1}^M\left (\sigma \left ( v \cdot a^{(m)} \right) - b^{(m)}\right )^2.
\ee
We first show the original classical algorithm, then an approximate classical algorithm, then the quantum algorithm.
The classical algorithm for the Sparsitron is given in Algorithm \ref{algSparsitron}.
\begin{figure}
\begin{algorithm}[H]
  \caption{Sparsitron \cite{Klivans2017} with risk approximation}
    \label{algSparsitron}
  \begin{algorithmic}[1]
    \Require{Parameter $\beta \in (0,1)$, norm $\lambda\geq 0$,
    training set $(x^{(t)},y^{(t)}) \in [-1,1]^N \times [0,1]$ for $t\in[T]$,
    training set $(a^{(m)},b^{(m)}) \in [-1,1]^N \times [0,1]$ for $m\in[M]$,
    flag $\in \{$original, approximate$\}$.}
    \State $w^{(1)} \gets \vec 1 / N$.
    \For{$t = 1 \textrm{ to } T$}
    \State $p^{(t)} \gets  \frac{w^{(t)}}{\Vert w^{(t)} \Vert_1}$.
   \State $l^{(t)} \gets \frac{1}{2} \left(\vec 1 +\left (\sigma\left (\lambda p^{(t)} \cdot x^{(t)}\right ) - y^{(t)}\right ) x^{(t)}\right)$.
    \State $w^{(t+1)} \gets w^{(t)}  \odot \beta^{l^{(t)}}$.
    \If {flag = original}
     \For{$m=1 \textrm{ to } M$}
    \State $z^{(t,m)} \gets p^{(t)} \cdot a^{(m)}$.
    \EndFor
    \ElsIf {flag = approximate}
    \State Prepare sampling data structure for $p^{(t)}$  via Fact \ref{theoremSampling}.
     \For{$m=1 \textrm{ to } M$}
    \State $z^{(t,m)}  \gets$ Estimate $p^{(t)} \cdot a^{(m)}$ to accuracy $\frac{\epsilon}{16\lambda}$ with success probability $1-\frac{\delta}{MT}$.
    \EndFor
    \EndIf
        \State $\widehat \varepsilon^{(t)} \gets \frac{1}{M} \sum_{m=1}^M\left (\sigma\left (\lambda z^{(t,m)} \right ) - b^{(m)}\right )^2$.
    \EndFor
     \Ensure {$v=\lambda p^{(t')}$ for $t' =\arg \min_{t\in[T]}\widehat \varepsilon^{(t)}$. }
  \end{algorithmic}
\end{algorithm}
\end{figure}
The algorithm consists of one main loop that goes over the first set of training examples and is equivalent to the Hedge algorithm.
For each training example $x^{(t)}$, a prediction vector $\lambda p^{(t)}$ is constructed.
From this vector one can compute the predicted label for the training example by using the activation function as $\sigma\left (\lambda p^{(t)} \cdot x^{(t)}\right)$.
A measure of the quality of the prediction is given by $\sigma \left (\lambda p^{(t)} \cdot x^{(t)}\right ) - y^{(t)}$ using the given label $y^{(t)}$ from the training set.
In the loop, a loss vector $l^{(t)} = \frac{1}{2} \left (\vec 1 + \left (\sigma \left (\lambda p^{(t)} \cdot x^{(t)}\right ) - y^{(t)}\right ) x^{(t)} \right ) \in [0,1]^N$ is computed, which has  entries close to $1/2$ in the case of correct prediction. In the last step of the loop, the second training set is used to compute the empirical risk $\widehat \varepsilon\left (\lambda p^{(t)}\right)$ to determine the quality of each prediction vector $p^{(t)}$.
The algorithm finally returns the vector $v=\lambda p^{(t')}$ which performs best on this set.
The provable learning guarantee and the run time is summarized in the following theorem. Plugging in $T$ and $M$, the run time can be expressed as $\tOrd{N \frac{\lambda^2}{ \epsilon^4} \log^2\frac{1}{\delta}}$.
\begin{theorem}[Sparsitron \cite{Klivans2017}]\label{theoremSparsitron}
Let $\mathcal D$ be a distribution on $[-1,1]^N \times \{0,1\}$ where for $(X,Y) \sim \mathcal D$, $\mathbbm E\left [Y|X = x\right] = \sigma(w \cdot x)$ for a non-decreasing 1-Lipschitz function $\sigma : \mathbbm R \to [0, 1]$. Suppose that $\Vert w\Vert_1\leq \lambda$ for a known $\lambda \geq 0$. Let $\epsilon,\delta \in (0, 1)$.  Given $T+M =\Ord{\lambda^2 \log (N/\delta \epsilon)/ \epsilon^2}$ independent samples from $\mathcal D$ and $\beta = 1-\sqrt{\log N /T}$, Algorithm \ref{algSparsitron} with flag ``original" produces a vector $v \in \mathbbm R^N$ such that with probability at least $1-\delta$,
\be
\varepsilon(v) \leq \epsilon.
\ee
The run time of the algorithm is $\Ord{N\times T\times M }$, where $T = \Ord{\lambda^2 \log (N/\delta \epsilon)/\epsilon^2}$ and $M = \Ord{ \log(T/\delta)/\epsilon^2}$. Moreover, the algorithm can be run in an online manner.
\end{theorem}
While the sample complexity of this algorithm is near-optimal \cite{SW12}, it is interesting to investigate classical run time improvements.
To this end, we now discuss our approximate Sparsitron algorithm, i.e., Algorithm \ref{algSparsitron} with the ``approximate" flag.
The difference to the original algorithm is that most inner products are estimated instead of computed exactly.

The provable learning guarantee for the approximate algorithm follows from the original work but relies on a few additional ideas.
First,
Fact \ref{theoremSampling} discusses the construction of a data structure to sample from a probability vector $p$. Given this data structure,  inner products $ p\cdot x$ for $x\in[-1,1]^N$ can be determined efficiently to additive accuracy $\epsilon$ and success probability $1-\delta$. The corresponding result is Lemma \ref{lemmaSampleInner} in Appendix \ref{appendixSampling}.
Since here $\Vert x \Vert_{\max} \leq 1$, the run time for a single inner product estimation is $\Ord{\frac{\Vert x \Vert_{\max}^2}{\epsilon^2} \log \frac{1}{\delta}}=\Ord{\frac{1}{\epsilon^2} \log \frac{1}{\delta}}$.
A scaling with $1/\epsilon^2$ is obtained, in contrast to quantum estimation which scales with $1/\epsilon$.

Second, as mentioned before the empirical risk is an approximation to the true risk as only $M$ training examples are used. The inner product estimation leads to an approximate empirical risk. For any vector $v \in [-1,1]^N$ and given samples $(a^{(m)},b^{(m)}) \in [-1,1]^N \times [0,1]$ for $m\in[M]$ from $\mathcal D$, let $z^{(m)}$ be the estimates of the inner product $v \cdot a^{(m)}$. The approximate empirical risk is defined as
\be \label{eqRiskApprox}
\widetilde {\varepsilon} (v) := \frac{1}{M} \sum_{m=1}^M\left (\sigma \left ( z^{(m)} \right) - b^{(m)}\right )^2.
\ee
We can use this approximate risk to bound the true risk via the triangle inequality. Third, each inner product estimation is probabilistic, hence we bound the overall success probability of the algorithm with a union bound together with the probabilistic steps of the original algorithm.

 As we now show, the run time of the approximate classical algorithm is about $\tOrd{T(N +\frac{\lambda^4 M }{\epsilon^2} \log \frac{1}{\delta})}$. The multiplicative updates and maintaining of a sampling data structure still cost $\tOrd{N}$.
Plugging in $T$ and $M$, the run time can be expressed as
$\tOrd{N \frac{\lambda^2}{\epsilon^2} \log \frac{1}{\delta} +\frac{\lambda^6 }{\epsilon^6} \log^3 \frac{1}{\delta}}$. In some ranges of parameters, this run time can be considered an improvement over the original Sparsitron which has a run time of $\tOrd{N \frac{\lambda^2}{ \epsilon^4} \log^2\frac{1}{\delta}}$.
\begin{theorem}[Approximate Sparsitron] \label{theoremSparsitronApprox}
With the same assumptions on $\mathcal D$, $\sigma$, $\epsilon$, $\delta$, $\lambda$, $T$, $M$, and $\beta$ as in Theorem \ref{theoremSparsitron},
Algorithm \ref{algSparsitron} with flag ``approximate" produces a vector $v \in \mathbbm R^N$ such that with probability at least $1-\delta$,
\be
\varepsilon(v) \leq \epsilon.
\ee
The run time of the algorithm is $\tOrd{T\left (N +\frac{M \lambda^4}{\epsilon^2} \log \frac{1}{\delta}\right)}$. Again, the algorithm can be run in an online manner.
\end{theorem}
\begin{proof}
In the original work \cite{Klivans2017}, due to the guarantees of the Hedge algorithm, it was derived that for the true risk and all computed $p^{(t)}$ it holds that
\be
\min_{t\in[T]} \varepsilon\left (\lambda p^{(t)}\right)  \leq  \epsilon.
\ee
We can easily let $\epsilon \to \epsilon/2$.
Our algorithm selects $v= \lambda p^{(t')}$ for $t' =\arg \min_{t\in[T]} \widetilde \varepsilon(\lambda p^{(t)})$, where $\widetilde \varepsilon$ is the empirical risk estimated from the imprecise inner products. We need to show that for $v$ we have a similar guarantee.

The closeness of the true and the empirical risk is achieved as in the original work by choosing $M = C\log(T/\delta)/\epsilon^2$, with a constant $C$,  such that for all $t\in[T]$ with probability $1-\delta$
\be
\left \vert \varepsilon\left(\lambda p^{(t)}\right) -  \widehat \varepsilon\left (\lambda p^{(t)}\right) \right \vert \leq \frac{\epsilon}{8}.
\ee
In addition, the error of inner products used for computing the empirical risk is set to $\left \vert  z^{(t,m)}- p^{(t)} \cdot a^{(m)} \right \vert \leq \frac{\epsilon}{16\lambda}$, hence the induced error in the empirical risk is bounded as
\be
\left \vert \widetilde \varepsilon\left (\lambda p^{(t)}\right) - \widehat \varepsilon\left(\lambda p^{(t)}\right) \right\vert &\leq& \frac{2 \lambda }{M}  \sum_{m=1}^M \left \vert z^{(t,m)} - p^{(t)} \cdot a^{(m)} \right \vert \nonumber \\
 &\leq&  \frac{\epsilon}{8},
\ee
since $b^{(m)}\in[0, 1]$ and the Lipschitz constant of $x^2$ on $[-1,1]$ is $2$.
Hence we also have that
\be
\left \vert \widetilde \varepsilon\left(\lambda p^{(t)}\right) - \varepsilon\left(\lambda p^{(t)}\right) \right\vert \leq \frac{\epsilon}{4}.
\ee
Therefore, the true risk of $v$ can be bounded as
\be \label{eqRiskBound}
\varepsilon(v) &\leq& \frac{\epsilon}{4}  +  \widetilde  \varepsilon( v) = \frac{\epsilon}{4} + \min_{t\in[T]} \widetilde \varepsilon\left(\lambda p^{(t)}\right) \nonumber \\
&=& \frac{\epsilon}{4} + \min_{t\in[T]}\left( \widetilde \varepsilon\left(\lambda  p^{(t)}\right) - \varepsilon\left(\lambda  p^{(t)}\right) + \varepsilon\left(\lambda  p^{(t)}\right)\right )\nonumber \\ &\leq& \frac{\epsilon}{2} + \min_{t\in[T]} \varepsilon\left(\lambda  p^{(t)}\right) \leq \epsilon.
\ee
Next, we discuss  the run time.
Computing the loss vector and maintaining the $ p^{(t)}$ sampling data structure costs $\tOrd{N}$. Computing the multiplicative update costs
$\Ord{N}$.
Estimating the $M$ inner products for risk estimation at every step to accuracy $\epsilon/(16 \lambda)$ with success probability $1-\frac{\delta}{MT}$ costs $\tOrd{M \frac{\lambda^2}{\epsilon^2} \log \frac{MT }{\delta}}$.
Hence the total run time is $\tOrd{T \left(N + M \frac{\lambda^2}{\epsilon^2} \log \frac{MT }{\delta}\right)}$ which can be simplified to $\tOrd{T \left(N + \frac{M \lambda^4}{\epsilon^2} \log \frac{1}{\delta}\right)}$.
The success probability of the inner loop is $\left (1-\frac{\delta}{MT}\right)^M \geq 1- \frac{\delta}{T}$.
The success probability over all steps is $\left (1-\frac{\delta}{T}\right)^{T}  \geq 1-\delta$.
Together with the success probability of the martingale estimation and the risk estimation of the original algorithm this gives a total success probability of the algorithm of at least $1-3\delta$. For the theorem statement, let $\delta \to \delta /3$.
\end{proof}
To achieve this run time, we were able to compute the first inner product $p^{(t)} \cdot x^{(t)}$ exactly and only estimate the $M$ inner products $p^{(t)} \cdot a^{(m)}$. In the quantum case, we will also estimate the first inner product, which increases the amount  of error analysis, as will be shown now.
\section{Quantum Sparsitron}

In this section, we construct a quantum algorithm for the Sparsitron. The algorithm is again based on quantum minimum finding, see Lemma \ref{lemmaMin}, amplitude amplification and estimation, see Lemma \ref{lemmaU}, and inner product estimation, see  Lemma \ref{lemmaInnerProduct2}, all in Appendix \ref{appendixAmp}.
Similar to the quantum Hedge algorithms above, the core idea is to never explicitly store the weight vector $w^{(t)}$. Rather, norms and inner products are estimated and stored.  In the iteration, access to these quantities and the new training datum allow to prepare a new loss oracle and a new weight quantum state. This  state preparation can then in turn be used to compute the new norm and inner products. We can expect to obtain a quantum speedup in the dimension $N$. On the other hand, we do not expect a quantum speedup in the number of samples $T$ as the provable learning guarantees are classical and invoke the Hedge algorithm. In fact, we obtain again a worsening of the performance in $T$. If  $\lambda$ and $1/\epsilon$ are  ${\rm poly} \log N$, this worsening is however tolerable as $T$ is then also $ {\rm poly} \log N$.

We first specify the input model. Here, we assume quantum access to the training data. The access model can be turned into an online setting by providing sequential access to the unitaries.
\begin{oracle}[Training sets]\label{inputQuantumSparsitron}
Let $j\in [N]$, $t\in [T]$, and $m \in [M]$.
Assume $\Ord{1}$ bits are sufficient to store $x_j^{(t)}$ and $a_j^{(m)}$. Assume to be given access to $T$ unitaries $U_{\rm train 1}^{(t)}$ and $M$ unitaries $U_{ \rm train 2}^{(m)}$ on $\Ord{\log N }$ qubits that perform the operations $\ket j \ket{\bar 0} \to \ket j \ket{ x_j^{(t)} }$ and $\ket j \ket{\bar 0} \to \ket j \ket{ a_j^{(m)}}$, respectively.
\end{oracle}
The same discussion regarding the $\ket {\bar 0}$ state as in Data Input \ref{oracle1} applies.
The data access allows to arithmetically compute the desired losses in quantum superposition.
\begin{lemma}[Loss quantum circuits] \label{lemmaSparsitronLossUnitaries}
Given Input \ref{inputQuantumSparsitron}
and classical access to the  numbers $\lambda \geq 0$, $h \in[-1,1]$ and $y \in \mathbbm [0,1]$.
For $t\in [T]$, the quantum operation $\ket j \ket{\bar 0} \to \ket j \ket{\frac{1}{2} \left (1 + \left (\sigma \left ( \lambda h\right) - y\right) x_j^{(t)}\right)}$  for $j\in [N]$ can be constructed on $\Ord{\log N }$ qubits, where the result is encoded to constant additive accuracy. The run time is $\Ord{1}$.  We denote these quantum circuits by $U_{\tilde l^{(t)}}$.
\end{lemma}
\begin{proof}
Compute $z:= \sigma(\lambda h) - y$ classically.
Use the quantum access to
$\ket j \ket{ x_j^{(t)}}$ in superposition.
Then we compute $\ket j \ket{ x_j^{(t)}} \ket{ \frac{1}{2} \left ( 1+ z x_j^{(t)} \right ) }$,
using the well-known quantum circuits for basic arithmetic operations. Uncompute the
second register via another query to obtain  $\ket j \ket{ \frac{1}{2} \left ( 1+ z x_j^{(t)}\right ) }$. 
\end{proof}
Hence, we are able to construct the loss unitaries, which in the previous quantum Hedge algorithm were assumed to be given in Input \ref{oracle1}. With these loss unitaries, we can compute the weights via Lemma \ref{lemmaUpdateBit} and perform minimum finding and $\ell_1$-norm estimation, as before. Computing the  inner products is the core step in the Sparsitron. Consider the inner product $h^{(t)} := \frac{w^{(t)}}{\Vert w^{(t)} \Vert} \cdot x^{(t)}$. Instead of this inner product, we estimate a shifted inner product because the $x_j^{(t)}$ are $[-1,1]$ and there can be cancelation effects which make the inner product zero or very close to zero.  Lemma \ref{lemmaInnerProduct2} discusses the quantum estimation of the inner product.
As a byproduct this lemma also provides  $ w^{(t)}_{\max}$ and an estimate of
$ \left\Vert \frac{ w^{(t)}}{ w_{\max}^{(t)} } \right \Vert_1$, hence we do not describe these steps separately.

Note that it is not important that the weights $w^{(t)}$ follow exactly the original Sparsitron weights. It is only important that we have a final guarantee from the Hedge algorithm. If the inner product estimations are accurate enough we only obtain a small additional error to the Hedge error bound.
We note related work which analyzes the regret bounds with noisy estimates of the gradients in the bandit setting \cite{FKM2005}.

We proceed with the algorithm  for the Quantum Sparsitron, see Algorithm  \ref{algSparsitronQuantum}.
\begin{figure*}
\begin{minipage}{\linewidth}
\begin{algorithm}[H]
  \caption{Quantum Sparsitron}
    \label{algSparsitronQuantum}
  \begin{algorithmic}[1]
    \Require{Error $\epsilon\in (0,1)$, probability $\delta \in(0,1)$, parameter $ \beta \in (0,1)$, norm $\lambda\geq 0$, quantum access to training set $(x^{(t)}, y^{(t)}) \in [-1,1]^N \times [0,1]$ for $t\in[T]$ and training set
     $(a^{(m)}, b^{(m)}) \in [-1,1]^N \times [0,1]$ for $m\in[M]$.
     }
    \For{$t = 1 \textrm{ to } T$}
    \State Construct unitary for $w^{(t)} = w^{(1)} \beta^{\tilde l^{(1)}}\cdots \beta^{\tilde l^{(t-1)}}$  via Lemma \ref{lemmaUpdateBit} using unitaries $\{U_{\tilde l^{(t')}} : t' \in [t-1] \}$ from Lemma \ref{lemmaSparsitronLossUnitaries}.
     \State $ h^{(t)} \gets$  Estimate $\frac{w^{(t)}}{ \Vert w^{(t)} \Vert_1} \cdot x^{(t)}$ to additive accuracy $ \frac{\epsilon}{8\lambda^2}$ with success probability $1-\frac{\delta}{2T}$ via Lemma \ref{lemmaInnerProduct2}. From this subroutine, also store $w^{(t)}_{\max}$ and the estimate of $ \left\Vert \frac{ w^{(t)}}{ w_{\max}^{(t)} } \right \Vert_1$ .
      \For{$m=1 \textrm{ to } M$}
     \State $z^{(t,m)} \gets$ Estimate $\frac{ w^{(t)} }{\Vert w^{(t)}\Vert_1} \cdot a^{(m)}$ to additive accuracy $\frac{\epsilon}{16 \lambda}$ with success probability $1-\frac{\delta}{2MT}$ via Lemma \ref{lemmaInnerProduct2}.
    \EndFor
    \State $\widetilde \varepsilon^{(t)} \gets \frac{1}{M} \sum_{m=1}^M \left(\sigma \left(\lambda  {z^{(t,m)}}\right ) - b^{(m)}\right)^2$.
    \State Construct unitary $U_{l^{(t)}}$ that prepares $\ket j \ket {\tilde l^{(t)}_j}$ with $\tilde l^{(t)} = \frac{1}{2}\left ( \vec 1 +\left(\sigma\left (\lambda { h^{(t)}}\right ) - y^{(t)}\right) x^{(t)}\right)$ for the next step.
    \EndFor
    \State $t' =\arg \min_{t\in[T]}\widetilde \varepsilon^{(t)}$.
     \Ensure {$\left ({h^{(1)}}, \cdots, {h^{(t')}} , \Gamma^{(t')}, w^{(t')}_{\max} \right)$. }
  \end{algorithmic}
\end{algorithm}
\end{minipage}
\end{figure*}
The output of the algorithm are inner product estimates and a norm estimate. The output is not a full classical vector, which would take
$\Ord{N}$ time and space to write down. The inner products estimates allow the preparation of a quantum state proportional to the desired vector $q$ from which one can take samples or compute inner products with other quantum states. Using $T$ and $M$, the run time is given by $\tOrd{\frac{\lambda^6 \sqrt{N}}{\epsilon^7 }\log^4 \left(\frac{1}{\delta} \right) }$, compared to the run time of
$\tOrd{N \frac{\lambda^2}{\epsilon^2} \log \frac{1}{\delta} +\frac{\lambda^6 }{\epsilon^6} \log^3 \frac{1}{\delta}}$ of the approximate Sparsitron and the run time of $\tOrd{N \frac{\lambda^2}{ \epsilon^4} \log^2\frac{1}{\delta}}$ of the original Sparsitron. The statement is as follows.

\begin{theorem}[Quantum Sparsitron]\label{theoremQuantumSparsitron}
Let the same assumptions on $\mathcal D$, $\sigma$, $\epsilon$, $\delta$, $\lambda$, and $\beta$ hold as in Theorem \ref{theoremSparsitron}.  Given $T+M = \Ord{\lambda^2\log (N/\delta \epsilon)/ \epsilon^2}$ independent samples from $\mathcal D$ accessed via Data Input \ref{inputQuantumSparsitron},
Algorithm \ref{algSparsitronQuantum} returns $\left ({h^{(1)}}, \cdots, {h^{(t')}} , \Gamma^{(t')}, w^{(t')}_{\max} \right)$, i.e, inner product estimates, a norm estimate, and a maximum weight for some $t'\in[T]$.
The run time of the algorithm to obtain this output is $\tOrd{\frac{\lambda^2 T^2 M \sqrt{N}}{\epsilon }\log \left(1/\delta \right) }$,
where $M = \Ord{ \log(T/\delta)/\epsilon^2}$.
Again, the algorithm can be run in an online manner.
Given this output of Algorithm  \ref{algSparsitronQuantum}, there exists a vector $q\in \mathbbm R^N$ such that its coordinates $q_j$ can be constructed separately in time $\tOrd{T}$ and  $q$ satisfies with probability at least $1-\delta$ that
\be
\varepsilon(q) \leq \epsilon.
\ee
In addition, an approximation to the quantum state $ \ket q = \sum_{j=1}^N \sqrt{ q_j/\Vert q\Vert_1} \ket j$ can be prepared in time $\tOrd{T \sqrt N \log\left(1/\delta \right)\log\left(1/\xi \right)}$ with success probability at least $1-\delta$, where the accuracy is $\xi \in(0,1)$.
\end{theorem}
\begin{proof}

\emph{Correctness}.
We generalize the analysis of \cite{Klivans2017} to the case where the inner products needed in the Sparsitron algorithm are only estimated to some accuracy (instead of being computed exactly). Recall that the loss vector at every step is given by
\be \label{eqLossEst}
\tilde l^{(t)} = \frac{1}{2} \left (\vec 1 +\left (\sigma \left( \lambda h^{(t)} \right) - y^{(t)}\right) x^{(t)}\right).
\ee
Define the random variable
\be
Q^{(t)} :=  p^{(t)} \cdot \tilde l^{(t)} - w \cdot \tilde l^{(t)}/\lambda.
\ee
For this random variable, we can establish three useful facts, generalizing the original work with respect to the estimated inner products.
The first fact is that since the weights are updated with the loss vector $\tilde l^{(t)}$, the resulting probability vector $ p^{(t)}$ satisfies a regret bound.
From Theorem \ref{thmRegret} it follows that
\be \label{eqthmhedge}
\sum_{t=1}^T p^{(t)} \cdot \tilde l^{(t)}\leq \min_{j\in[N]}  \tilde L_j + \sqrt{2T \log N} + \log N,
\ee
where $\tilde L_j = \sum_{t=1}^T \tilde l_j^{(t)}$.
The second fact is that the sequence $Q^{(t)}$ is not too far from its expectation value.
As in the original work, for the bounded  martingale difference sequence
$ Q^{(t)} - \mathbbm E_{(x^{(t)},y^{(t)})}\left[ Q^{(t)} \Big \vert (x^{(1)},y^{(1)}), \cdots, (x^{(t-1)},y^{(t-1)}) \right]$,
the Azuma-Hoeffding inequality implies with probability at least $1-\delta$ that
\be \label{eqAzuma}
\sum_{t=1}^T &&\mathbbm E_{(x^{(t)},y^{(t)})}\left[ Q^{(t)} \Big \vert (x^{(1)},y^{(1)}), \cdots, (x^{(t-1)},y^{(t-1)}) \right] \leq \nonumber \\
&& \sum_{t=1}^T Q^{(t)} + \Ord{\sqrt {T \log (1/\delta)}}.
\ee
The third fact involves lower-bounding the expectation value of $Q^{(t)}$, taking into account that the inner products are estimated. Consider first
\be
\zeta^{(t)} :=  p^{(t)} \cdot  x^{(t)} - w \cdot  x^{(t)} /\lambda,
\ee
for which
\be \label{eqZetaBound}
 \left \vert \zeta^{(t)} \right \vert \leq  \left \Vert  p^{(t)} - w/\lambda \right \Vert_1 \leq 2,
\ee
since $\left \Vert x^{(t)} \right\Vert_{\max} \leq 1$.
For the expectation value of $Q^{(t)}$ we obtain
\begin{eqnarray} \label{eqexpvalueq}
&\mathbbm E_{\left (x^{(t)},y^{(t)}\right)} \left[ Q^{(t)} \Big \vert \left (x^{(1)},y^{(1)}\right), \cdots, \left (x^{(t-1)},y^{(t-1)}\right) \right] = \nonumber \\
&= \frac{1}{2} \mathbbm E_{(x^{(t)},y^{(t)})}\left[ \zeta^{(t)} \left (\sigma\left ( \lambda h^{(t)} \right) - \sigma\left (w\cdot x^{(t)}\right ) \right)\right],
\end{eqnarray}
using that $  p^{(t)} \cdot \vec 1 = 1$ and $w \cdot \vec 1/\lambda = 1$.
Let $h^{(t)}$ be an estimate of $p^{(t)} \cdot x^{(t)}$ with error $\left \vert h^{(t)} -  p^{(t)} \cdot x^{(t)} \right \vert \leq \epsilon_{px}$.
From the $1$-Lipschitz property, it is easy to see that $\left \vert \sigma \left (\lambda h^{(t)} \right ) -\sigma\left (\lambda{ p^{(t)} \cdot x^{(t)}}\right ) \right\vert \leq \lambda \epsilon_{px} $.
Using this fact and Eq.~(\ref{eqZetaBound}), we obtain
\be \label{eqErrorTemp}
\zeta^{(t)} \sigma\left( \lambda h^{(t)}\right) &\geq& \zeta^{(t)} \sigma \left(\lambda p^{(t)} \cdot x^{(t)}\right) - 2 \lambda \epsilon_{px}.
\ee
Using Eq.~(\ref{eqErrorTemp}) in Eq.~(\ref{eqexpvalueq}), we obtain the lower bound
\be \label{eqExpLower}
&\mathbbm E_{\left (x^{(t)},y^{(t)}\right)} \left[ Q^{(t)} \Big \vert \left (x^{(1)},y^{(1)}\right), \cdots \right] \nonumber \\
&\quad \geq  \frac{1}{2\lambda} \mathbbm E_{(x^{(t)},y^{(t)})}\left[ \left (\sigma \left ( \lambda  p^{(t)} \cdot x^{(t)}\right)  - \sigma\left(w\cdot x^{(t)} \right) \right )^2\right] \nonumber \\ & -  \lambda \epsilon_{px}.
\ee
Here we have used that for all $a,b \in \mathbbm R$ that $(a-b)(\sigma(a)-\sigma(b)) \geq (\sigma(a)-\sigma(b))^2$.

These three facts related to $Q^{(t)}$ are now combined to derive the guarantee of the algorithm.
The first term in Eq.~(\ref{eqExpLower}) is by definition the risk $\frac{1}{2\lambda}  \varepsilon\left(\lambda  p^{(t)}\right)$. Hence, for the risk we have
\be \frac{\varepsilon\left(\lambda p^{(t)}\right ) }{2\lambda}  &\leq&
\mathbbm E_{(x^{(t)},y^{(t)})}\left[ Q^{(t)} \Big \vert (x^{(1)},y^{(1)}), \cdots \right] \nonumber \\
&& + \lambda \epsilon_{px}.
\ee
Using the result derived from the Azuma-Hoeffding inequality, Eq.~(\ref{eqAzuma}), we obtain with probability at least $1-\delta$ that
\be
\frac{1}{2\lambda} \sum_{t=1}^T \varepsilon\left(\lambda p^{(t)}\right) &\leq&\sum_{t=1}^T Q^{(t)} + \Ord{\sqrt{T\log (1/\delta)}} \nonumber \\
&& + \lambda \epsilon_{px} T.
\ee
Using the Hedge regret bound Eq.~(\ref{eqthmhedge}), we can bound
\be
\sum_{t=1}^T Q^{(t)} &\leq& \min_{j\in[N]}  \tilde L_j + \sqrt{2T \log N}\nonumber  + \log N  - \sum_{t=1}^T
\frac{w \cdot \tilde l^{(t)}}{\lambda} \\
&\leq& \sqrt{2T \log N} + \log N.
\ee
The second inequality follows from $\min_{j\in[N]}  \tilde L_j - \sum_{t=1}^T w \cdot \tilde l^{(t)}/\lambda \leq 0$ since $\lambda = \Vert w\Vert_1$.
From the upper bound for the sum, we obtain an upper bound for the minimum element since $\min_{t\in [T]} \varepsilon\left(\lambda  p^{(t)}\right) \leq \frac{1}{T} \sum_{t=1}^T \varepsilon\left(\lambda  p^{(t)}\right)$.
Similar to the original work,
setting $T > C' \lambda^2 \log(N/\delta)/\epsilon^2$ with a constant $C'$ and $\epsilon_{px} =\frac{ \epsilon}{8\lambda^2} $ we obtain
\be
\min_{t\in[T]} && \varepsilon\left(\lambda p^{(t)}\right)  \leq  \nonumber \\ &&\leq  \Ord{\lambda} \frac{\sqrt{2T \log N} + \log N + \sqrt{T \log 1/\delta} }{T}
\nonumber + 2 \lambda^2 \epsilon_{px}  \\
&&\leq \frac{\epsilon}{4} +  \frac{\epsilon}{4} = \frac{\epsilon}{2}.
\ee
Note that the algorithm selects $v= \lambda  p^{(t')}$ for $t' =\arg \min_{t\in[T]} \widetilde \varepsilon\left(\lambda p^{(t)}\right)$, where $\widetilde \varepsilon$ is the empirical risk estimated from the imprecise inner products. 
The final risk bound is achieved as in the original work and the approximate Algorithm \ref{algSparsitron}. As in Eq.~(\ref{eqRiskBound}), the true risk of $v$ is bounded as
\be
\varepsilon(v) \leq \epsilon.
\ee
\emph{Run time.}
We discuss the inner product estimation and the total run time.
Fix  $t \in[T]$.
Assume we are given the Input \ref{inputQuantumSparsitron}, the Sparsitron inner product estimates up to time $t-1$, ${h^{(1)}}, \cdots, {h^{(t-1)}}$, and the corresponding unitaries $U_{\tilde l^{(1)}},\cdots,U_{\tilde l^{(t-1)}}$ from Lemma \ref{lemmaSparsitronLossUnitaries}.
Together with the weight computation Lemma \ref{lemmaUpdateBit} and quantum access to the training data, construct the operation that computes $\ket j  \ket{ w_j^{(t)}} $ by basic arithmetic operations and uncomputing unnecessary registers.

Using Lemma \ref{lemmaInnerProduct2}, we obtain an estimate $ {h^{(t)}}$ of the inner product $\frac{ w^{(t)}}{\Vert  w^{(t)}\Vert_1} \cdot x^{(t)}$. The additive accuracy is $ \frac{\epsilon}{8\lambda^2}$ and the success probability is $1-\frac{\delta}{2T}$.
This step uses in total $\tOrd{\frac{\lambda^2 T \sqrt{N}}{\epsilon} \log\left (\frac{1}{\delta} \right)}$ quantum gates.

For the $M$ inner products for the risk estimation at every step,
we again use Lemma \ref{lemmaInnerProduct2} to obtain an estimate $z^{(t,m)}$ of the inner product $\frac{ w^{(t)}}{\Vert  w^{(t)}\Vert_1} \cdot a^{(m)}$. The additive accuracy is $\frac{\epsilon}{16 \lambda}$ and the success probability is $1-\frac{\delta}{2MT}$.
This step uses in total $\tOrd{\frac{\lambda T \sqrt{N}}{\epsilon} \log\left (\frac{1}{\delta} \right)}$ quantum gates.

The total cost of the algorithm is thus  $\tOrd{\frac{\lambda^2 T^2\sqrt{N}}{\epsilon }\log \left(\frac{1}{\delta} \right) + \frac{ \lambda  T^2 M \sqrt{N} }{\epsilon} \log (\frac{1}{\delta} )}$ which can be simplified to $\tOrd{\frac{\lambda^2 T^2 M \sqrt{N}}{\epsilon }\log \left(\frac{1}{\delta} \right) }$.
The success probability of the inner loop is $(1-\frac{\delta}{2MT})^{M} \geq 1- \frac{\delta}{2T}$. The  success probability of all the probabilistic steps in the algorithm is $\left (1-\frac{\delta}{2T}\right)^{2T}\geq 1-\delta$.
This leads to a total success probability of at least $1-3\delta$. Finally, let $\delta \to \delta/3$.

\emph{Output.}
From the output of the algorithm, the construction of a single element $q_j$ of a classical vector takes time $\tOrd{T}$. That is because
\be
q_j = \lambda \frac{ \beta^{\sum_{t=1}^{t'-1} \frac{1}{2} \left (1 +\left (\sigma\left ( {h^{(t)}}\right) - y^{(t)}\right) x_j^{(t)}\right)} } { w^{(t')}_{\max} \Gamma^{(t')}},
\ee
where the computation of the sum takes time $\Ord{T}$.
For the quantum state preparation, use Lemma \ref{lemmaU} with $u_j = q_j/q_{\max}$, where $q_{\max} = \max_j q_j$, and the efficient computability of $q_j$ in $\Ord{T}$. The maximum finding can be done as before.
By Lemma \ref{lemmaU} (iii), preparing $ \ket q = \sum_{j=1}^N \sqrt{ q/\Vert q\Vert_1} \ket j$ takes a run time of $\tOrd{T \sqrt N \log\left(\frac{1}{\delta} \right) \log\left(\frac{1}{\xi} \right)}$.
\end{proof}

We continue with remarks on the practicality of the output of the quantum algorithm.
The main scenario is the classification of a new sample. Let the new sample be  $x_{\rm new} \in [-1,1]^N$ and quantum access to the sample be given similar to Data Input \ref{inputQuantumSparsitron}. With the same amplitude estimation methods as above, estimate the inner product $q\cdot x_{\rm new}$ of the new sample and the vector $q$. Then apply the Lipschitz function to the estimate as $\sigma (q\cdot x_{\rm new})$ to obtain an estimate of the label $y_{\rm new} \in [0, 1]$. These steps only add a proportional overhead to the run time of the Sparsitron and provide classification information on new samples with correctness guarantees.
One can envision other scenarios. First, we may not need classical knowledge of all the elements $q_j$ but rather only a small known subset of them. Hence we compute the $q_j$ only on that subset, which can be done efficiently.
Second, we may use the quantum state $\ket q$ to obtain knowledge about the large elements of the vector. If $\ket q$ is sparse, then after a small number of measurements in the computational basis we obtain the positions $j$ of the large elements and use this knowledge to compute them.
Finally, we may have measurement operators $O_i$ providing information about global properties of $w$. We can measure these operators on the state $\ket q$ to provide information about the true $w$ more efficiently than a classical algorithm.
In summary, the algorithm may find practical use in these and other scenarios. Similar to many quantum algorithms for machine learning, one main obstacle for practical use is quantum data access. Functioning large-scale quantum data input devices remain yet to be developed, hence our algorithm may show practical use only in the longer term once such devices are available.

Finally, we  show the application of the Quantum Sparsitron to Ising models as a corollary. Please refer to the beginning of Section \ref{sectionSparsitron} for a brief introduction to the problem. Recall that the width of an Ising model is defined as $\lambda(A,\theta)=\max_i \left (\sum_j \vert A_{ij} \vert + \vert \theta_i \vert \right)$, see Eq.~(\ref{eqProbGibbs}) for the definition of $A$ and $\theta$.
The algorithm is a combination of the algorithm in \cite{Klivans2017} with the Quantum Sparsitron discussed above.
\begin{cor}[Quantum Learning of Ising models]\label{corollaryIsing}
Given an $N$-variable Ising model with width $\leq \lambda$ for $\lambda \geq 0$. Given quantum query access to the entries of the samples from the Ising model.  Given $\epsilon, \delta \in (0,1)$, and $T = \Ord{\lambda \exp(\Ord{\lambda})/\epsilon^4)\log(N/\delta\epsilon)}$ independent samples from the Ising distribution, there exists a quantum algorithm that produces classical inner product estimates and norms such that every element of a matrix $A^\ast$ can be computed in time $\tOrd{T}$. For the matrix $A^\ast$ it holds that $\Vert A-A^\ast \Vert_{\max} \leq\epsilon$ with probability at least $1-\delta$.
The run time of the algorithm is $\tOrd{\frac{\lambda^2 T^2 M N^{3/2}}{\epsilon }\log \left(\frac{1}{\delta} \right) }$, where $M = \Ord{ \log(T/\delta)/\epsilon^2}$.
Quantum states of the columns/rows of the matrix $A^\ast$ with $\epsilon$ distance can be prepared in time $\tOrd{\frac{T \sqrt N}{\epsilon}}$.
Again, the algorithm can be run in an online manner.
\end{cor}
\begin{proof}
Apply the Quantum Sparsitron $\Ord{N}$ times for each node of the Ising model, hence the run time is
$\tOrd{\frac{\lambda^2 T^2 M N^{3/2}}{\epsilon }\log \left(\frac{T}{\delta} \right) }$.
At this cost, the Quantum Sparsitron has the same guarantee as the classical Sparsitron.
It was shown in \cite{Klivans2017} that the obtained guarantees for the GLMs can be translated to the guarantee $\Vert A-A^\ast \Vert_{\max} \leq\epsilon$ for the matrix of the Ising model.
Hence we obtain the guarantee for the Ising model learning.
The rows/columns of $A^\ast$ can be reconstructed as in Theorem \ref{theoremQuantumSparsitron}, either classically element by element or as quantum states.
\end{proof}
The classical algorithm by Klivans and Meka~\cite{Klivans2017} has a run time of $\Ord{N^2 T}$  with the same near-optimal sample complexity. Corollary \ref{corollaryIsing} can hence be considered an improvement in the $N$ dependency, which comes at the price of a worsening of the run time in $T$ and $\epsilon$.  
\section{Discussion and Conclusion}

In the main part of this work, we have presented a quantum machine learning algorithm with both provable learning guarantee and provable quantum speedup over the best known classical algorithm.
The starting point is a classical algorithm called the Sparsitron, which is a
dimensionally sample-optimal algorithm for generalized linear models under modest assumptions. Generalized linear models have a large range of applications, including logistic and Poisson regression, and they appear also in a variety of problems such as the learning of Ising models and Markov Random Fields (MRFs).

The run time of our quantum algorithm, the Quantum Sparsitron, shows a speedup polynomial in the dimension of the problem and a slowdown in the error dependency, while the sample complexity remains the same as for the classical algorithm.
The setting here is the standard quantum gate model, i.e.,  many logical quantum bits with the physical errors being kept under control via error correction. In addition, we assume the availability of unitaries (oracles) which provide access to the training examples. The training examples can be given via efficient quantum circuits or via classical data collected from sampling the true distribution and quantum RAM access to these data.
The main quantum subroutines are the well-known amplitude amplification and estimation algorithms, which are here applied in a way that the provable learning guarantee and success probability of the classical algorithm are preserved.
Due to the use of amplitude amplification and oracles, the algorithm can be considered more far-term in nature, requiring significantly more resources than the presently available 50-100  noisy qubits.

The optimization problem here is in principle non-convex. The Lipschitz function defining the generalized linear model can be a non-convex function (such as the \emph{quasi}-convex sigmoid function), which leads to hardness results for learning even a single neuron \cite{AHW96,Klivans2017}. In the sense that the classical Sparsitron solves this non-convex problem, the quantum algorithm solves the same non-convex problem. While many of the recent quantum algorithms are for convex problems, such as LPs and SDPs \cite{BS17,vAG18}, this work can be seen as an extension of the same underlying quantum techniques to non-convex problems.

On the classical side, we have shown that the Sparsitron run time (but not the sample complexity) can be sped up via inner product estimation techniques. Randomized linear algebra is a well-studied area which has recently also found application in the discussion of dimension-efficient classical algorithms for various problems considered for quantum machine learning \cite{Tang2018,Tang2018_2,Chia2018_linear,Gilyen2018_regression,Chia2019_SDP,Chia2019_framework}. Our work relates to these results in the sense that we have started at a near-optimal classical machine learning algorithm.
It may be interesting to take  near-optimal versions of \cite{Tang2018,Tang2018_2,Chia2018_linear,Gilyen2018_regression,Chia2019_SDP,Chia2019_framework} as starting points  and exhibit quantum speedups for the various problems.

\section{Acknowledgements}

This work was supported by the Singapore National Research Foundation, the Prime Minister's Office, Singapore, the Ministry of Education, Singapore under the Research Centres of Excellence programme under research grant R 710-000-012-135, and Baidu-NUS Research Project Nr.~2019-03-07.
This work was also partially funded by QuantERA ERA-NET Cofund project QuantAlgo and the ANR project ANR-18-CE47-0010 QUDATA.

\appendix

\section{Classical sampling} \label{appendixSampling}

We mention a result for the efficient sampling from a probability vector.
\begin{fact} [$\ell_1$-sampling \cite{Vos91,Wal74}]\label{theoremSampling}
Given an $N$-dimensional probability vector $p$. There exists a data structure to sample an index $j\in[N]$ with probability $p_j$ which can be constructed in time $\tOrd{N}$. One sample can be obtained in time $\tOrd{1}$.
\end{fact}
Regarding this result, Ref.~\cite{Vos91} shows an algorithm with preparation in $\Ord{N}$ and sampling in $\Ord{1}$ assuming constant time operations for addition, comparison, and random number generation, among others. However, storing and processing pointers $j\in[N]$ takes $\Ord{\log N}$ bits and operations, hence we take the slightly worse $\tOrd{N}$ and $\tOrd{1}$, respectively.

We discuss the sampling of inner products. The proof is standard and adapted from
\cite{Tang2018} which shows the $\ell_2$-sampling case.
\begin{lemma}[Inner product estimation] \label{lemmaSampleInner}
Let $\epsilon,\delta \in (0,1)$. Given query access to $x \in [-1,1]^N$ and $\ell_1$-sampling access to an $N$-dimensional probability vector $p$. We can determine $ p\cdot x$ to additive error $\epsilon$ with success probability at least $1-\delta$ with $\Ord{\frac{\Vert x \Vert_{\max}^2}{\epsilon^2} \log \frac{1}{\delta}}$ queries and samples, and $\tOrd{\frac{\Vert x \Vert_{\max}^2}{\epsilon^2} \log \frac{1}{\delta}}$ time complexity.
\end{lemma}
\begin{proof}
Define a random variable $Z$ with outcome $x_j$ with probability $p_j$.
Note that $\mathbbm E[Z] = \sum_j p_j x_j = p \cdot x$.
Also, $\mathbbm V[Z] \leq \sum_j x_j^2 p_j \leq \Vert x \Vert_{\max}^2$.
Take the median of $6 \log 1/\delta$ evaluations of the mean of $9/(2\epsilon^2)$ samples of $Z$
to be within $\epsilon \sqrt{\mathbbm V[Z]} \leq \epsilon \Vert x \Vert_{\max}$ of $p \cdot x$
with probability at least $1-\delta$ in $\Ord{\frac{1}{\epsilon^2} \log \frac{1}{\delta}}$ queries.
\end{proof}

\section{Quantum subroutines}
\label{appendixAmp}
\begin{proof} [Proof of Lemma \ref{lemmaUpdateBit}]
With a computational register involving $\Ord{T}$ ancilla qubits for the losses, perform $\ket j \ket{\bar 0} \to \ket j \ket{l^{(1)}_j} \dots \ket{l^{(t-1)}_j}\ket{\bar 0} \to \ket j \ket{l^{(1)}_j} \dots \ket{l^{(t-1)}_j} \ket{\sum_{t'=1}^{t-1} l_j^{(t')}}$ to sufficient accuracy, using the oracles. Uncomputing the loss registers via additional queries leads to the result. With the quantum circuits for basic arithmetic operations, the operations for $\ket j \ket{w^{(t)}_j}$ and $\ket j \ket{w^{(t)}_j/w^{(t)}_{\max}}$ can also be prepared.
\end{proof}
In the following, for the data access to the vector $u$, consider the discussion regarding the $\ket {\bar 0}$ state given after Data Input \ref{oracle1}.
\begin{lemma}[Quantum minimum finding \cite{Durr1996}]
\label{lemmaMin}
Given quantum access to a vector $u \in [0,1]^N$ via the operation $\ket j \ket{\bar 0} \to \ket j \ket{ u_j}$ on $\Ord{\log N}$ qubits, where $u_j$ is encoded to additive accuracy $\Ord{1/N}$. Then, we can find the minimum $u_{\min} = \min_{j\in[N]}  u_j$ with success probability $1-\delta$ with $\Ord{\sqrt N \log \left (\frac{1}{\delta}\right) }$ queries and $\tOrd{\sqrt N  \log \left( \frac{1}{\delta}\right )}$ quantum gates.
\end{lemma}
The minimum finding can be turned straightforwardly into a maximum finding algorithm.
Next, we state the results for estimating the $\ell_1$-norm of a vector and preparing states encoding the square root of the vector elements.
\begin{lemma}[Quantum state preparation and norm estimation] \label{lemmaU}
Let $\eta >0$.
Given a non-zero vector $u \in [0,1]^N$, with $\max_j u_j = 1$. Given quantum access to $u$ via the operation $\ket j \ket{\bar 0} \to \ket j \ket{ u_j}$ on $\Ord{\log N + \log 1/ \eta}$ qubits, where $u_j$ is encoded to additive accuracy $\eta$. Then:
 \begin{enumerate}[(i)]
\item There exists a unitary operator that prepares the state
$\frac{1}{\sqrt{N}}  \sum_{j=1}^N \ket j  \left( \sqrt{ u_j } \ket{0} + \sqrt{1- u_j} \ket{1} \right)$ with two queries and number of gates $\Ord{\log N + \log 1/\eta}$.
Denote this unitary by $U_\chi$.
\item Let $\epsilon >0$ such that $\eta \leq \epsilon/ (2N)$ and $\delta \in (0,1)$. There exists a quantum algorithm that  provides an estimate $\Gamma_u$ of the $\ell_1$-norm
$\Vert u \Vert_1$  such that
$\left \vert \Vert u \Vert_1 - \Gamma_u\right \vert \leq \epsilon \Vert u \Vert_1$,
with probability at least $1-\delta$. The algorithm requires $\Ord{\frac{\sqrt{ N}}{\epsilon} \log(1/\delta)}$  queries and $\tOrd{\frac{\sqrt{ N }}{\epsilon} \log\left (1/\delta\right)}$ gates.
\item Let  $ \xi \in (0,1]$ such that $\eta \leq \xi /4N$ and $\delta \in(0,1)$. An approximation
$\ket{ \tilde p} =  \sum_{j=1}^N   \sqrt{ \tilde p_j }\ket j $ to the state
$\ket{u} := \sum_{j=1}^N   \sqrt{ \frac{u_j}{\Vert u \Vert_1}}\ket j  $ can be prepared with probability $1-\delta$, using
$\Ord{\sqrt{N} \log (1/\delta)}$ calls to the unitary of (i) and $\tOrd{\sqrt{N} \log (1/\xi)\log \left (1/\delta \right)}$  gates.
The approximation in $\ell_1$-norm of the probabilities is $\left \Vert   \tilde p -  \frac{u}{\Vert u\Vert_1} \right\Vert_1 \leq \xi$.
 \end{enumerate}
\end{lemma}
\begin{proof}
For (i),  prepare a uniform superposition of all $\ket j$ with $\Ord{\log N}$ Hadamard gates. With the quantum query access, perform
\be
\frac{1}{\sqrt{N}} \sum_{j=1}^N &\ket j& \ket {\bar 0} \to \frac{1}{\sqrt{N}}  \sum_{j=1}^N \ket j  \ket{ u_j} \ket { 0} \\
&\to& \frac{1}{\sqrt{N}}  \sum_{j=1}^N \ket j  \ket{ u_j}  \left( \sqrt{ u_j} \ket{0} + \sqrt{1-u_j } \ket{1} \right). \nonumber
\ee
The steps consist of an oracle query and a controlled rotation. The rotation is well-defined as $ u_j \leq 1$ and costs $\Ord{ \log 1/\eta}$ gates. Then uncompute the data register $\ket{  u_j } $ with another oracle query.

For (ii),
define a unitary $\mathcal U = U_\chi \left(\mathbbm 1 - 2 \ket{\bar 0}\bra{\bar 0}\right) \left(U_\chi\right)^\dagger$, with $U_\chi$ from (i). Define another unitary by $\mathcal V = \mathbbm 1-\mathbbm 1 \otimes \ket{0} \bra{0}$.
Using $K$ applications of $\mathcal U$ and $\mathcal V$, Amplitude Estimation \cite{Brassard2002} allows to provide an estimate $\tilde a$ of the quantity
$a = \frac{\Vert u \Vert_1}{N }$
to accuracy $\vert \tilde a - a \vert \leq 2 \pi \frac{\sqrt{a(1-a)} }{K} + \frac{\pi^2}{K^2}$.
Following \cite{Apeldoorn2017}, take
$
K> \frac{6 \pi }{\epsilon} \sqrt{ N} ,
$
 which obtains
 \be
\vert \tilde a - a \vert
 &\leq& \frac{\pi}{K}\left( 2\sqrt{a}+ \frac{\pi}{K} \right)
 < \frac{\epsilon_1}{6 }  \sqrt{ \frac{1}{N }  } \left( 2\sqrt{a}+  \frac{\epsilon }{12}\sqrt{ \frac{1}{N } } \right) \nonumber \\
 &\leq&  \frac{\epsilon}{6 }  \sqrt{ \frac{1}{N}  }  \left( 3 \sqrt{a} \right)
 = \frac{\epsilon \sqrt{\Vert u \Vert_1} }{2 N}.
\ee
Since $\Vert u \Vert_1 \geq 1$ by assumption, we have
$ \vert \tilde a - a \vert \leq \frac{\epsilon \Vert u \Vert_1 }{2 N}$.
Also, there is an inaccuracy arising from the additive error $\eta$ of each $u_j$. As it was assumed that
$\eta\leq \epsilon / (2N)$,  the overall multiplicative error $\epsilon$ is obtained for the estimation.
For performing a single run of amplitude estimation with $K$ steps,
we require
$\Ord{K} =\Ord{\frac{\sqrt{ N }}{\epsilon} }$
queries to the oracles and
$\Ord{\frac{\sqrt{ N }}{\epsilon} \left(\log N + \log (N / \epsilon)  \right) }$
gates.

For (iii), rewrite the state from (i)
as
\be \sqrt{ \frac{{\Vert  u \Vert_1} }{N}}  &\sum_{j=1}^N&    \sqrt{ \frac{u_j  }{{\Vert  u \Vert_1}}} \ket j \ket{0}  \\ &+&  \sqrt{1-\frac{{\Vert  u \Vert_1}}{N }}  \sum_{j=1}^N  \sqrt{\frac{1- u_j }{N -{\Vert  u \Vert_1}} } \ket j \ket{1}.\nonumber
\ee
Now amplify the $\ket 0$ part using Amplitude Amplification \cite{Brassard2002} via the exponential search technique without knowledge of the normalization, to prepare
$\sum_{j=1}^N \ket j   \sqrt{ \frac{u_j}{\Vert u \Vert_1}}$ with success probability $1-\delta$.
The amplification requires $\Ord{\sqrt{ \frac{N }{ \Vert u \Vert_1}}\log (1/\delta)}= \Ord{\sqrt N \log (1/\delta)}$ calls to the unitary of (i), as ${\Vert u \Vert_1} \geq 1$. The gate complexity derives from the gate complexity of (i).
Denote the  $\eta$-additive approximation to $u_j$ by $\tilde u_j$,
and evaluate the $\ell_1$-distance of the probabilities.
First, $ \left \vert \Vert u\Vert_1 - \Vert \tilde u\Vert_1\right \vert \leq N \eta$.
One obtains $\left \Vert   \tilde p -  \frac{u}{\Vert u\Vert_1} \right \Vert_1 =\left \Vert  \frac{\tilde u}{{\Vert \tilde u\Vert_1}}  -  \frac{u}{\Vert u\Vert_1} \right \Vert_1 \leq \sum_j \left \vert \frac{\tilde u_j}{{\Vert \tilde  u\Vert_1}}  -  \frac{u_j}{{\Vert \tilde u\Vert_1}} \right \vert + \sum_j \left \vert \frac{ u_j}{{\Vert \tilde u\Vert_1}}  -  \frac{u_j}{\Vert u\Vert_1}\right \vert  \leq \frac{N \eta} {{\Vert \tilde u\Vert_1}} + \frac{ N \eta } {{\Vert \tilde u\Vert_1}}  $. We also obtain $\frac{1}{ {\Vert \tilde u\Vert_1} } \leq \frac{1}{ \Vert u\Vert_1 -N\eta } \leq \frac{2}{ \Vert u\Vert_1}$ for $\eta \leq  \Vert u\Vert_1 / 2 N $. Since $\eta \leq \Vert u\Vert_1 \xi/(4N)$, the distance is $ \left \Vert   \tilde p -  \frac{u}{\Vert u\Vert_1} \right \Vert_1 \leq \xi$ as desired.
\end{proof}

\begin{lemma} [Quantum inner product estimation with relative accuracy] \label{lemmaInnerProduct1}
Let $\epsilon,\delta \in(0,1)$.
Given quantum access to two vectors $u,v \in [0,1]^N$, where $u_j$ and $v_j$ are encoded to additive accuracy $\eta= \Ord{1/N}$. Then,
an estimate $I$ for the inner product can be provided such that $\vert I - u\cdot v /\Vert u\Vert_1 \vert \leq \epsilon\  u\cdot v /\Vert u\Vert_1$ with success probability $1-\delta$.
This estimate is obtained with $\Ord{\frac{ \sqrt{N} }{\epsilon} \log \left (\frac{1}{\delta} \right )  }$ queries and $\tOrd{\frac{ \sqrt{N} }{\epsilon} \log \left (\frac{1}{\delta} \right )  }$ quantum gates.
\end{lemma}
\begin{proof}
Via Lemma \ref{lemmaMin}, determine $u_{\max}$ with success probability $1-\delta$ with $\Ord{\sqrt N \log \frac{1}{\delta}}$ queries and $\tOrd{\sqrt N \log \left( \frac{1}{\delta}\right )}$ quantum gates.
Apply  Lemma \ref{lemmaU} with the vector $ \frac{ u}{ u_{\max}}$ to obtain an estimate $\Gamma_u$ of the norm $\left \Vert  \frac{ u}{ u_{\max}} \right \Vert_1$ to relative accuracy $\epsilon_u= \epsilon/2$ with success probability $1-\delta$.
This estimation takes $\Ord{\frac{ \sqrt{N} }{\epsilon} \log \left (\frac{1}{\delta} \right )  }$ queries and $\tOrd{\frac{ \sqrt{N} }{\epsilon} \log \left (\frac{1}{\delta} \right )  }$ quantum gates.

Define the vector $z$ with $z_j = u_j v_j$.
Via Lemma \ref{lemmaMin}, determine $ z_{\max}$ with success probability $1-\delta$ with $\Ord{\sqrt N \log \frac{1}{\delta}}$ queries and $\tOrd{\sqrt N \log \left( \frac{1}{\delta}\right )}$ quantum gates.
If $ z_{\max} = 0$ up to numerical accuracy, the estimate is $I = 0$ and we are done.
Otherwise,
apply  Lemma \ref{lemmaU} with the vector $ \frac{ z}{ z_{\max}}$ to obtain an estimate $\Gamma_z$ of the norm $\left \Vert  \frac{ z}{ z_{\max}} \right \Vert_1$ to relative accuracy $\epsilon_u = \epsilon/2$ with success probability $1-\delta$.
This estimation takes $\Ord{\frac{ \sqrt{N} }{\epsilon} \log \left (\frac{1}{\delta} \right )  }$ queries and $\tOrd{\frac{ \sqrt{N} }{\epsilon} \log \left (\frac{1}{\delta} \right )  }$ quantum gates.

With Lemma \ref{lemmaErrorRatio}, we have
\be
\left \vert \frac{\Gamma_z}{\Gamma_u} - \frac{u_{\max}}{z_{\max}}\frac{u\cdot v}{\Vert u\Vert_1} \right \vert &\leq& \frac{u_{\max}}{z_{\max}} \frac{u\cdot v }{\Vert u\Vert_1 } \frac{\epsilon_z + \epsilon_u}{ (1-\epsilon_u)} \\
&\leq& 2\epsilon \frac{u_{\max}}{z_{\max}} \frac{u\cdot v }{\Vert u\Vert_1 },
\ee
since $\epsilon_u < 1/2$.
Set
\be
I= \frac{z_{\max}}{u_{\max}} \frac{\Gamma_z}{\Gamma_u},
\ee
and we have $\vert I - u\cdot v /\Vert u\Vert_1 \vert \leq 2\epsilon\  u\cdot v /\Vert u\Vert_1$.
The total success probability of  the four probabilistic steps is at least $1-4\delta$ via a union bound. Choosing $\epsilon \to \epsilon/2$ and $\delta \to \delta/4$ leads to the result.
\end{proof}

\begin{lemma} [Quantum inner product estimation with additive accuracy] \label{lemmaInnerProduct2}
Let $\epsilon,\delta \in(0,1)$.
Given quantum access to a non-zero vector $u \in [0,1]^N$ and another vector $v \in [-1,1]^N$, where $u_j$ and $v_j$ are encoded to additive accuracy $\eta= \Ord{1/N}$. Then,
an estimate $I$ for the inner product can be provided such that $\vert I - u\cdot v / \Vert u\Vert_1 \vert \leq \epsilon$ with success probability $1-\delta$.
This estimate is obtained with $\Ord{\frac{ \sqrt{N} }{\epsilon} \log \left (\frac{1}{\delta} \right )  }$ queries and $\tOrd{\frac{ \sqrt{N} }{\epsilon} \log \left (\frac{1}{\delta} \right )  }$ quantum gates.
\end{lemma}
Note that as a byproduct, the value $u_{\max}$ and an estimate of $\Vert u /u_{\max} \Vert_1$ with relative accuracy $\epsilon$ can be provided with probability at least $1-\delta$.
\begin{proof}
Via Lemma \ref{lemmaMin}, determine $\Vert u\Vert_{\max}$ with success probability $1-\delta$ with $\Ord{\sqrt N \log \frac{1}{\delta}}$ queries and $\tOrd{\sqrt N  \log \left( \frac{1}{ \eta}\right) \log \left( \frac{1}{\delta}\right )}$ quantum gates.
Apply  Lemma \ref{lemmaU} with the vector $ \frac{ u}{ u_{\max}}$ to obtain an estimate $\Gamma_u$ of the norm $\left \Vert  \frac{ u}{ u_{\max}} \right \Vert_1$ to relative accuracy $\epsilon_u= \epsilon/2$ with success probability $1-\delta$.
This estimation takes $\Ord{\frac{ \sqrt{N} }{\epsilon} \log \left (\frac{1}{\delta} \right )  }$ queries and $\tOrd{\frac{ \sqrt{N} }{\epsilon} \log \left (\frac{1}{\delta} \right )  }$ quantum gates.

Similarily, consider the vector $z$ with elements $z_j := u_j \left (v_j+3\right) \in [0,4]$.
Determine $\Vert z\Vert_{\max}$ with success probability $1-\delta$ with $\Ord{\sqrt N \log \frac{1}{\delta}}$ queries and $\tOrd{\sqrt N  \log \left( \frac{1}{\delta}\right )}$ quantum gates.
Apply  Lemma \ref{lemmaU} with the vector $z/ z_{\max}$ to obtain an estimate $\Gamma_z$ of the norm $\Vert z/ z_{\max} \Vert_1$ to relative accuracy $\epsilon_z =  \epsilon/2$ with success probability $ 1-\delta$.
This estimation takes $\Ord{\frac{ \sqrt{N} }{\epsilon} \log \left (\frac{1}{\delta} \right )  }$ queries and $\tOrd{\frac{ \sqrt{N} }{\epsilon} \log \left (\frac{1}{\delta} \right )  }$.

The exact quantities are related via
\be
\frac{u\cdot v}{\Vert u\Vert_1}  = \frac{z_{\max} }{u_{\max} } \frac{ \Vert \frac{z}{ z_{\max}} \Vert_1 } { \Vert \frac{ u}{ u_{\max}} \Vert_1}  - 3.
\ee
Considering the estimator $I = \frac{z_{\max} }{u_{\max} } \frac{ \Gamma_z } { \Gamma_u}  - 3$, from Lemma \ref{lemmaErrorRatio}, we have
\be
\left  \vert I -  \frac{u\cdot v}{ \Vert u\Vert_1 } \right \vert &=& \frac{z_{\max} }{u_{\max} } \left \vert   \frac{ \Gamma_z } { \Gamma_u} - \frac{ \Vert \frac{z}{ z_{\max}} \Vert_1 } { \Vert \frac{ u}{ u_{\max}} \Vert_1}  \right \vert \\ &\leq&
\frac{ \epsilon_u + \epsilon_{z}}{1-  \epsilon_u}  \frac{ \Vert z \Vert_1 } { \Vert u\Vert_1} \leq 8 \epsilon. \nonumber
\ee
In the last steps we have used
that
\be
 \frac{ \Vert z \Vert_1 } { \Vert u\Vert_1} \equiv \frac{\sum_j u_j  (v_j+3)}{\sum_j u_j } \leq\frac{4 \sum_j  u_j }{\sum_j u_j } = 4,
\ee
and $\epsilon_u < 1/2$.

All steps together take $\Ord{\frac{\sqrt N}{\epsilon} \log \frac{1}{\delta}}$ queries and $\tOrd{\frac{\sqrt N}{\epsilon}  \log \left( \frac{1}{\delta}\right )}$ gates.
The total success probability of all the probabilistic steps is at least $1-4\delta$ via a union bound. Choosing $\epsilon \to \epsilon/8$ and $\delta \to \delta/4$ leads to the result.

\end{proof}

\begin{lemma}\label{lemmaErrorRatio}
Let $\tilde a$ be an estimate of $a>0$ such that
$\vert \tilde a- a \vert \leq \epsilon_a a$.
with $\epsilon_a \in (0,1)$.
Similarly, let $\tilde b$ be an estimate of $b>0$ and $\epsilon_b \in (0,1)$ such that
$\vert \tilde b - b \vert \leq \epsilon_b b$.
Then the ratio $a/b$ is estimated to relative error
$
 \left \vert \frac{\tilde a}{\tilde b} -  \frac{a}{b}  \right\vert \leq \left ( \frac{\epsilon_a + \epsilon_b}{1-\epsilon_b} \right) \frac{a}{b}$.
\end{lemma}
\begin{proof}
Note that
$ b -  \tilde b \leq \vert \tilde b - b\vert \leq \epsilon_b b$, from which we
deduce $\frac{1}{ \tilde b} \leq \frac{1}{ b (1-\epsilon_b)}$.
In addition,
$
 \left \vert \frac{\tilde a}{\tilde b} -  \frac{a}{b}  \right\vert =
  \left \vert \frac{\tilde a b - a \tilde b}{\tilde b b}  \right\vert = \left \vert \frac{\tilde a b - ab + ab - a \tilde b}{\tilde b b}  \right\vert = \left \vert \frac{\tilde a  - a}{\tilde b} + \frac{a}{\tilde b} \frac{b - \tilde b}{ b}  \right\vert
  \leq \left \vert \frac{\tilde a  - a}{\tilde b} \right\vert+  \frac{a}{\tilde b}  \left \vert \frac{b - \tilde b}{ b}  \right\vert \leq \frac{\epsilon_a a + \epsilon_b a }{\tilde b}   \leq  \frac{a}{b}\frac{\epsilon_a +\epsilon_b}{ (1-\epsilon_b)}.
$
\end{proof}

\bibliographystyle{apsrev}
\bibliography{Mult,Qfin,Subm}

\end{document}